\documentclass[a4paper,UKenglish,cleveref, autoref, thm-restate]{lipics-v2021}

\usepackage{tikz,pgffor}
\usepackage{ifthen}
\usetikzlibrary{arrows,backgrounds,calc,patterns.meta}
\usetikzlibrary{automata,positioning,arrows,shapes,snakes,math,arrows.meta,decorations.pathmorphing}
\tikzstyle{min}=[thick,circle,draw,minimum size=1.4em,inner sep=0em,text centered]
\tikzstyle{dec}=[circle,draw,fill,minimum size=.8ex,inner sep=0em]
\tikzstyle{state}=[draw,thick,minimum width=7mm,minimum height=5mm,rounded corners,text centered]
\tikzstyle{tran}=[-stealth,rounded corners,thick]

\usepackage{thmtools} 
\usepackage{multicol}

\usepackage{xcolor}
\makeatletter
\def\squiggly{\bgroup \markoverwith{\textcolor{red}{\lower3.5\p@\hbox{\sixly \char58}}}\ULon}
\makeatother

\newcommand{\N}{\mathbb{N}}
\newcommand{\R}{\mathbb{R}}
\newcommand{\Q}{\mathbb{Q}}

\newcommand{\A}{\mathcal{A}}

\newcommand{\B}{\mathcal{B}}
\newcommand{\F}{\mathcal{F}}

\newcommand{\M}{\mathcal{M}}

\newcommand{\X}{\mathcal{X}}

\newcommand{\calR}{\mathcal{R}}

\newcommand{\EXPTIME}{\textbf{EXPTIME}}
\newcommand{\PSPACE}{\mbox{\rm\bf PSPACE}}

\newcommand{\AP}{\textit{AP}}
\newcommand{\run}{\textit{Run}}
\newcommand{\NP}{\textbf{NP}}
\newcommand{\coNP}{\mathbf{coNP}}

\newcommand{\tran}[1]{\xrightarrow{\makebox[4em]{\footnotesize$#1$}}}
\newcommand{\ex}[1]{\langle #1 \rangle}

\newcommand{\Init}{\mathit{Init}}

\newcommand{\Simulate}{\mathit{Sim}}
\newcommand{\NewSimulate}{\mathit{NewSim}}

\newcommand{\Step}{\mathit{Step}}

\newcommand{\Zero}{\mathit{Zero}}

\newcommand{\NewZero}{\textit{NewZero}}

\newcommand{\Line}{\mathit{Line}}

\newcommand{\Area}{\mathcal{A}}
\newcommand{\Out}{\mathit{Out}}
\newcommand{\Ins}{\mathit{Ins}}
\newcommand{\slope}{\mathit{slope}}

\newcommand{\Succ}{\mathit{Succ}}
\newcommand{\USucc}{\mathit{USucc}}
\newcommand{\at}{\mathit{at}}
\newcommand{\INC}{\mathit{Inc}}
\newcommand{\DEC}{\mathit{Dec}}
\newcommand{\Prob}{\mathit{Prob}}

\renewcommand{\vec}[1]{\pmb{#1}}
\newcommand{\Struct}{\textit{Struct}}
\newcommand{\Decrement}{\textit{Decrement}}

\newcommand{\rSuc}{\textit{rsuc}}
\newcommand{\RSuc}{\textit{Rsuc}}
\newcommand{\RKSuc}{\textit{RKsuc}}
\newcommand{\Mark}{\textit{Mark}}
\newcommand{\Eligible}{\textit{Eligible}}
\newcommand{\Copy}{\textit{Copy}}
\newcommand{\UCopy}{\textit{UCopy}}
\newcommand{\FLambda}{\textit{Lambda}}
\newcommand{\update}{\textit{update}}
\newcommand{\Update}{\textit{Update}}
\newcommand{\Dec}{\textit{dec}}
\newcommand{\UDec}{\textit{UDec}}
\newcommand{\Inc}{\textit{inc}}
\newcommand{\UInc}{\textit{UInc}}
\newcommand{\conf}{\textit{conf}}
\renewcommand{\Form}{\textit{Form}}
\newcommand{\Rec}{\textit{Rec}}
\newcommand{\Recurrent}{\textit{Recurrent}}
\newcommand{\Sync}{\textit{Sync}}

\newcommand{\cv}{\vec{\gamma}}
\newcommand{\vv}[1]{\INC^{#1}(\vec{z})}

\newcommand{\fp}{Appendix}

\newtheorem{prop}{Property}

\newcommand*{\vp}{\varphi}

\newcommand*{\opx}{\operatorname{\pmb{\mathtt{X}}}}
\newcommand*{\opu}{\operatorname{\pmb{\mathtt{U}}}}
\newcommand*{\opf}{\operatorname{\pmb{\mathtt{F}}}}
\newcommand*{\opg}{\operatorname{\pmb{\mathtt{G}}}}

\newcommand*{\nat}{\mathbb{N}}

\newcommand*{\m}{\mathbb{P}}

\renewcommand{\paragraph}[1]{\medskip

\noindent\textbf{#1.}\quad }

\title{The Satisfiability and Validity Problems for Probabilistic Computational Tree  Logic are Highly Undecidable}

\titlerunning{PCTL Satisfiability and Validity} 

\author{Miroslav Chodil}{Faculty of Informatics, Masaryk University, Brno, Czechia}{miroslav@chodil.com}{0000-0002-8406-0443}{}
\author{Anton\'{\i}n Ku\v{c}era}{Faculty of Informatics, Masaryk University, Brno, Czechia}{tony@fi.muni.cz}{0000-0002-6602-8028}{}

\authorrunning{Chodil and Ku\v{c}era}
\Copyright{Miroslav Chodil and Anton\'{\i}n Ku\v{c}era}


\funding{The presented work received funding from the European Union’s Horizon Europe program under the Grant Agreement No.\ 101087529.}

\nolinenumbers 

\hideLIPIcs  

\EventEditors{Keren Censor-Hillel, Fabrizio Grandoni, Joel Ouaknine, and Gabriele Puppis}
\EventNoEds{4}
\EventLongTitle{52nd International Colloquium on Automata, Languages, and Programming (ICALP 2025)}
\EventShortTitle{ICALP 2025}
\EventAcronym{ICALP}
\EventYear{2025}
\EventDate{July 8--11, 2025}
\EventLocation{Aarhus, Denmark}
\EventLogo{}
\SeriesVolume{334}
\ArticleNo{163}

\begin{document}

\maketitle


\begin{abstract}
    The Probabilistic Computational Tree Logic (PCTL) is the main specification formalism for discrete probabilistic systems modeled by Markov chains. 
    Despite serious research attempts, the decidability of PCTL satisfiability and validity problems remained unresolved for 30~years. We show that both problems are \emph{highly undecidable}, i.e., beyond the arithmetical hierarchy. Consequently, there is no sound and complete deductive system for PCTL. 
\end{abstract}

\section{Introduction}
\label{sec-intro}

Algorithms for checking the satisfiability/validity of formulae in a given logic have received significant attention from the very beginning of computer science. The satisfiability/validity problems are closely related to each other because for most logics we have that $\varphi$ is valid iff $\neg\varphi$ is not satisfiable. The decidability of the validity problem for first-order logic, also known as the \emph{Entscheidungsproblem}, was posed by Hilbert and Ackermann in~1928 and aswered negatively by Church \cite{Church:Entscheidungsproblem} and Turing \cite{Turing:Entscheidungsproblem} in 1936. Later, Trakhtenbrot \cite{Trakhtenbrot-FOL-finitesat} demonstrated that the \emph{finite} satisfiability problem for first-order logic is also undecidable. Consequently, the \emph{finite validity} problem is not even semi-decidable, and hence there is no sound and complete deductive system proving all finite valid first-order formulae. For propositional logic, the satisfiability and validity are the canonical $\NP$ and $\coNP$ complete problems \cite{Cook:SAT-NP-complete}.

The satisfiability and validity problems have also been intensively studied for \emph{temporal logics} such as LTL, CTL, CTL$^*$, or the modal $\mu$-calculus, and their computational complexity has been fully classified (see, e.g., \cite{Emerson:temp-logic-handbook,Stirling:temp-logic-handbook}). Temporal logics are interpreted over \emph{transition systems}, i.e., directed graphs where the vertices correspond to the states of an abstract non-deterministic program, and the edges model the atomic computational steps. Even for the most expressive logic of the modal $\mu$-calculus \cite{Kozen:mu-calculus}, the satisfiability/validity problems are decidable and \EXPTIME-complete, and every satisfiable formula $\varphi$ has a model whose size is at most exponential in $|\varphi|$ (this is known as ``the small model property'' \cite{Kozen:mu-calculus-finite-model}). Furthermore, there is a sound and complete deductive system for the modal $\mu$-calculus \cite{Walukiewicz:mucalculus-complete-LICS}.

\emph{Probabilistic temporal logics}, such as PCTL or PCTL$^*$ \cite{HJ:logic-time-probability-FAC} are interpreted over discrete-time Markov chains 
modelling \emph{probabilistic programs}. 
Despite numerous research attempts resulting in partial positive results  (see Related work), the decidability of the satisfiability/validity problems for PCTL has remained unresolved for more than 30~years. In this paper, we prove that both problems are \emph{highly undecidable} (i.e., beyond the arithmetical hierarchy), even for a simple PCTL fragment where the path connectives are restricted to $\opf$ and $\opg$. Consequently, there is no sound and complete deductive system for PCTL.
\smallskip

\noindent
\textbf{Related work.}
Unlike non-probabilistic temporal logics, PCTL does not have the small model property. In fact, one can easily construct satisfiable PCTL formulae without \emph{any} finite model (see, e.g., \cite{BFKK:satisfiability}). Hence, the PCTL satisfiability/validity problems have also been studied in their finitary variants.

The first positive decidability results have been obtained for the \emph{qualitative fragment} of PCTL, where the range of admissible probability constraints is restricted to ${=}0$, ${>}0$, ${=}1$, and~${<}1$ (see Section~\ref{sec-prelim} for details). The satisfiability/validity for qualitative PCTL is \EXPTIME-complete, and the same holds for finite satisfiability/validity  \cite{LS:time-chance-IC,HS:Prob-temp-logic,KL:qPCTL-satisfiability}. Furhermore, a finite description of a model for a satisfiable qualitative PCTL formula is constructible in exponential time \cite{BFKK:satisfiability}. The underlying proof techniques are similar to the ones used for non-probabilistic temporal logics.

The finite satisfiability is decidable also for certain \emph{quantitative} PCTL fragments with general probability constrains \cite{KR:PCTL-unbounded,CHK:PCTL-simple,ChK:PCTL-quatitative-fragments-JCSS}. The underlying arguments are mostly based on establishing an upper bound on the size of a model for a satisfiable formula of the considered fragment (the existence of a model can then be expressed in the existential fragment of first-order theory of the reals, which is known to be in $\PSPACE$ \cite{Canny:Tarski-exist-PSPACE}).
More concretely, in~\cite{CHK:PCTL-simple} it is shown that every formula $\varphi$ of the \emph{bounded fragment} of PCTL, where the validity of~$\varphi$ in a state $s$ depends only on a bounded prefix of a run initiated in~$s$, has a bounded-size tree model. In \cite{KR:PCTL-unbounded}, several PCTL fragments based on $\opf$ and $\opg$ operators are studied. For each of these fragments, it is shown that every \emph{finite} satisfiable formula has a bounded-size model where every non-bottom SCC is a singleton. For some of these fragments, it is shown that every satisfiable formula has a finite model, which yields the decidability of the (general) satisfiability problem.  In \cite{ChK:PCTL-quatitative-fragments-JCSS}, more abstract decidability results about finite PCTL satisfiability based on isolating the progress achieved along a chain of visited SCCs are presented. A variant of the bounded satisfiability problem, where transition probabilities are restricted to $\{\frac{1}{2},1\}$, is proven \NP-complete in \cite{BFS:bounded-PCTL}. A recent result of \cite{ChK:PCTL-finite-sat-LICS} says that the \emph{finite} satisfiability problem for unrestricted PCTL is \emph{undecidable}. For a given Minsky machine, a PCTL formula~$\psi$ is constructed so that every reachable configuration of the machine corresponds to a special state of a model. Hence, $\psi$ is \emph{always} satisfiable, and $\psi$ has a \emph{finite} model iff the set of reachable configurations of the machine is finite. Hence, the construction does not imply the undecidability of general PCTL satisfiability, and some of the crucial tools used in \cite{ChK:PCTL-finite-sat-LICS} apply only to finite Markov chains (see Section~\ref{sec-geom} for more comments).
\smallskip 

\noindent
\textbf{Our contribution.} We show that the satisfiability problem for PCTL is hard for the $\Sigma_1^1$ level of the analytical hierarchy. Consequently, the PCTL validity problem is $\Pi_1^1$-hard. These results hold even for the PCTL fragment where only the path connectives $\opf,\opg$ are allowed. The proof is based on encoding the existence of a recurrent computation of a given non-deterministic two-counter machine $\M$ by an effectively constructible PCTL formula $\varphi_\M$. The overall structure of $\varphi_\M$ and the main underlying concepts are introduced in Section~\ref{sec-simulate}. In particular, in Section~\ref{sec-sim-step}, we identify a fundamental obstacle that needs to be overcome when constructing $\varphi_\M$. This is achieved by utilizing some concepts of convex geometry, and the solution is presented in Section~\ref{sec-geom}. The construction of $\varphi_\M$ is then described in more detail in Section~\ref{sec-formula}. Interestingly, our results imply that the \emph{finite} satisfiability of the $\opf,\opg$-fragment is also undecidable. This is a non-trivial refinement of the result achieved in~\cite{ChK:PCTL-finite-sat-LICS}, where the finite satisfiability is shown undecidable for PCTL formulae containing the $\opf$, $\opg$, $\opx$, and $\opf^k$ path connectives. 

The main \emph{technical} contributions are the following. Non-negative integer counter values are represented by \emph{pairs} of probabilities. More concretely, we fix a suitable $\vec{z} \in (0,1)^2$ representing zero, and then design a suitable function $\INC :(0,1)^2 \to (0,1)^2$ such that every $n \geq 0$ is represented by the pair $\INC^n(\vec{z})$, where $\INC^n$ denotes the $n$-fold composion of $\INC$. The function $\INC$ is a modification of the function $\sigma$ designed in \cite{ChK:PCTL-finite-sat-LICS}. Then, we prove that every (possibly \emph{infinite}) $T \subseteq (0,1)^2$ satisfying a certain closure property must be contained in the convex polytope determined by the set of points $\{\INC^n(\vec{z}) \mid n\geq 0\}$ (see Theorem~\ref{thm-area}). This is the crucial (novel) ingredient for overcoming the obstacles identified in Section~\ref{sec-sim-step}. The second major technical contribution is a (novel) technique for representing the points $\INC^n(\vec{z})$ by pairs of PCTL formulae of the $\opf,\opg$-fragment. 

The presented construction is non-trivial and technically demanding. We believe that, at least to some extent, this is caused by the inherent difficulty of the considered problem. We did our best to provide a readable sketch of the main steps in Section~\ref{sec-formula}, where we also explain the purpose and use of Theorem~\ref{thm-area}. Full technical details are in \fp.

\section{Preliminaries}
\label{sec-prelim}

The sets of non-negative integers, rational numbers, and real numbers are denoted by $\N$, $\Q$, and $\R$, respectively. 
We use $\vec{u},\vec{v},\vec{z},\ldots$ to denote the elements of $\R \times \R$. The first and the second components of $\vec{u}$ are denoted by $\vec{u}_1$ and $\vec{u}_2$. For a function $f : A \to A$, we use $f^n$ to denote the $n$-fold composition $f \circ \cdots \circ f$.


\subsection{The Logic PCTL}
The logic PCTL \cite{HJ:logic-time-probability-FAC} is obtained from the standard CTL (Computational Tree Logic \cite{Emerson:temp-logic-handbook}) by replacing the existential and universal path quantifiers with the probabilistic operator $P(\Phi) \bowtie r$. PCTL formulae are interpreted over Markov chains where every state $s$ is assigned a subset $v(s) \subseteq \AP$ of propositions valid in~$s$.
 
\begin{definition}[PCTL]
\label{def-pctl}
   Let $\AP$ be a set of atomic propositions. The syntax of PCTL state and path formulae is defined by the following abstract syntax equations:
   \[
   \begin{array}{lcl}
      \varphi & ~~::=~~ & a \mid \neg \vp \mid \vp_1 \wedge \vp_2 \mid P(\Phi) \bowtie r\\
      \Phi & ::= &\opx \vp \mid \vp_1 \opu \varphi_2 \mid \vp_1 \opu^{k} \varphi_2
   \end{array} 
   \]    
   Here, $a \in \AP$, ${\bowtie} \in \{{\geq}, {>}, {\leq},{<},{=},{\neq}\}$, $r \in [0,1]$ is a rational constant, and $k \in \N$.
\end{definition}

  
A \emph{Markov chain} is a triple $M = (S,\Prob,v)$, where $S$ is a finite or countably infinite set of \emph{states}, \mbox{$\Prob \colon S \times S \rightarrow [0,1]$} is a probability matrix, 
    and \mbox{$v \colon S \rightarrow 2^{\AP}$} is a \emph{valuation}. 
For $s,t \in S$, we say that $t$ is an \emph{immediate successor} of $s$ if \mbox{$\Prob(s,t) > 0$}.
A \emph{path} in $M$ is a finite sequence $w = s_0, \ldots ,s_n$ of states where $n \geq 0$ and 
\mbox{$\Prob(s_i,s_{i+1}) > 0$} for all $i <n$. We say that $t$ is \emph{reachable} from $s$ if there is a path from~$s$ to~$t$.
A \emph{run} in $M$ is an infinite sequence $\pi = s_0, s_1, \ldots$ of states such that every finite prefix of $\pi$ is a path in $M$. We also use $\pi(i)$ to denote the state $s_i$ of~$\pi$. 
 


For every path $w = s_0, \ldots ,s_n$, let $\run(w)$ be the set of all runs starting with~$w$, and let $\m(\run(w)) = \prod_{i=0}^{n-1} \Prob(s_i,s_{i+1})$. To every state $s$, we associate the probability space $(\run(s),\F_{s},\m_{s})$, where $\F_{s}$ is the \mbox{$\sigma$-field} generated by all $\run(w)$ where $w$ starts in $s$, and $\m_{s}$ is the unique probability measure obtained by extending $\m$ in the standard way (see, e.g., \cite{Billingsley:book}). 
The \emph{validity} of a PCTL state/path formula for a given state/run of $M$ is defined inductively as follows:
\[
\begin{array}{lcl}
  s \models a & \mbox{iff} & a \in v(s),\\
  s \models \neg\varphi & \mbox{iff} & s \not\models \varphi,\\
  s \models \vp_1 \wedge \vp_2 &  \mbox{iff} & s \models \vp_1 \mbox{ and } s \models \vp_2,\\
  s \models P(\Phi) \,{\bowtie}\, r  &  \mbox{iff} &  \m_{s}(\{ \pi \in \run(s) \mid \pi \models \Phi \}) \bowtie r,\\[1ex]
  \pi \models \opx \vp  &  \mbox{iff} & \pi(1) \models \vp \mbox{ for some } i \in \N,\\
\pi \models \vp_1 \opu \varphi_2 &  \mbox{iff} & \mbox{there is } j\geq 0 \mbox{ such that }
\pi(j) \models \vp_2
  \mbox{ and }
    \pi(i) \models \vp_1 \mbox{ for all } 0\leq i < j,\\
  \pi \models \vp_1 \opu^k \varphi_2 &  \mbox{iff} & \mbox{there is } 0 \leq j\leq k \mbox{ such that } \pi(j) \models \vp_2
     \mbox{ and }
   \pi(i) \models \vp_1 \mbox{ for all } 0\leq i < j.\\
\end{array} 
\]

We say that $M$ is a \emph{model} of $\varphi$ if $s \models \varphi$ for some state $s$ of $M$. The \emph{PCTL satisfiability problem} is the question of whether a given PCTL formula has a model.

The formulae $\textit{true},\textit{false}$ and the other Boolean connectives are defined using $\neg$ and $\wedge$ in the standard way. In the following, we abbreviate a formula of the form $P(\Phi) \bowtie r$ by omitting $P$ and adjoining the probability constraint directly to the topmost path operator of $\Phi$. For example, we write $\opx_{{=}1} \varphi$ instead of $P(\opx\varphi) = 1$. 
We also use $\opf_{\bowtie r} \vp$ and $\opf_{\bowtie r}^k \vp$ to abbreviate the formulae $\textit{true} \opu_{\bowtie r} \vp$ and $\textit{true} \opu^k_{\bowtie r} \vp$, respectively. Furthermore, $\opg_{\bowtie r} \vp$ abbreviates $\opf_{\not\bowtie r} \neg \vp$. Hence, $s \models \opg_{\bowtie r} \varphi$ iff 
$\m_{s}(\{ \pi \in \run(s) \mid \pi(i) \models \varphi \mbox{ for all } i \geq 0\}) \bowtie r$.
\smallskip

\subsection{Non-Deterministic $k$-Counter Machines}
%
Our results are obtained by constructing a PCTL formula encoding the existence of a recurrent computation in a given non-deterministic two-counter machine. However, the standard Minsky machines \cite{Minsky:book} are not particularly convenient for this purpose. More concretely, since every instruction modifies only \emph{one} of the counters, a faithful simulation requires a mechanism for ``copying'' the other counter value to the next configuration. In other words, we need to implement not only the test-for-zero, increment, and decrement instructions, but also the `no-change' operation. 

To avoid the need for implementing the `no-change' operation, we introduce a special variant of non-deterministic counter machines where every instruction performs an operation on \emph{all} counters simultaneously. 
Let $d \geq 1$. A non-deterministic $d$-counter machine $\M$ is a finite sequence $\Ins_1, \ldots, \Ins_m$ of instructions operating over non-negative counters $C_1,\ldots,C_d$, where every $\Ins_\ell$ takes the form
\begin{equation}
   \langle C_k{=}0\, {?}\ Z \mbox{ : } P\rangle:\ \update_1, \ldots, \update_d
   \label{M-ins}
\end{equation}
where $k \in \{1,\ldots,d\}$, $Z,P \subseteq \{1,\ldots,m\}$ are non-empty subsets of target instruction indexes, and $\update_1,\ldots,\update_d \in \{\Dec,\Inc\}$ are counter updates.  $\M$ is \emph{deterministic} if $Z$ and $P$ are singletons in all $\Ins_\ell$. We use $C_k^\ell$, $Z^\ell$, $P^\ell$, and $\update_1^\ell, \ldots,\update_d^\ell$ to denote the respective elements of $\Ins_\ell$.

The instruction~\eqref{M-ins} is executed as follows. First, the counter $C_k$ is tested for zero, and a successor label $\ell$ is chosen non-deterministically from either $Z$ or $P$, depending on whether the test succeeds or not, respectively. Then, $\update_1,\ldots,\update_d$ simultaneously update the current values of $C_1,\ldots,C_d$. Here, $\Inc$ increments the counter by one, and $\Dec$ either decrements the counter by one or leaves the counter unchanged, depending on whether its current value is positive or zero, respectively. After that, the computation continues with $\Ins_\ell$. Note that the zero test is used \emph{only} to determine the next instruction label.

More precisely, a \emph{configuration} of $\M$ is a tuple of the form $(\ell,c_1,\ldots,c_d)$ where $\ell \in \{1,\ldots,m\}$ is the current instruction index, and $c_1,\ldots,c_d \in \N$ are the current values of $C_1,\ldots,C_d$. A configuration $(\ell',c_1',\ldots, c_d')$ is a \emph{successor} of $(\ell,c_1,\ldots,c_d)$, written $(\ell,c_1,\ldots,c_d) \mapsto (\ell',c_1',\ldots,c_d')$, if the following conditions are satisfied, where
$\Ins_\ell ~\equiv~ \langle C_k{=}0\, {?}\ Z \mbox{ : } P \rangle:\ \update_1,\ldots \update_d$.
\begin{itemize}
  \item If $c_k = 0$ then $\ell' \in Z$; otherwise $\ell' \in P$.
  \item For all $j \in \{1,\ldots,d\}$ we have that $c_j'$ is equal to either $c_j {+} 1$ or $\max\{0,c_j{-}1\}$, depending on whether $\update_j$ is $\Inc$ or $\Dec$, respectively.  
\end{itemize}
A \emph{computation} of $\M$ is an infinite sequence 
$\conf_0,\conf_1,\ldots$ of configurations such that $\conf_0 = (1,0,\ldots,0)$ and $\conf_i \mapsto \conf_{i+1}$ for all $i \in \N$. A computation is \emph{bounded} if it contains only finitely many pairwise different configurations, and \emph{$\tau$-recurrent}, where $\tau \subseteq \{1,\ldots,m\}$, if it contains infinitely many configurations of the form $(i,c_1,\ldots,c_d)$ where $i \in \tau$.
 
In the next proposition, $\Sigma_1^0$ and $\Sigma_1^1$ denote the corresponding levels in the arithmetical and the analytical hierarchies, respectively. A proof is obtained by simulating the standard Minsky machines, see \fp.

\begin{restatable}{proposition}{Minsky}
\label{prop-twocounter}
   The problem of whether a given non-deterministic two-counter machine $\M$ has a 
   $\tau$-recurrent computation (for a given $\tau$) is $\Sigma_1^1$-hard. Furthermore, the problem of whether a given deterministic two-counter machine 
   $\M$ has a bounded computation is $\Sigma_1^0$-hard. 
\end{restatable}

\section{Simulating Non-Deterministic Two-Counter Machines by PCTL}
\label{sec-simulate}

In this section, we give a high-level description of our undecidability proof, together with some preliminary intuition about the underlying techniques. The presented ideas are elaborated in subsequent sections.

Let $\M$ be a non-deterministic two-counter machine with~$m$ instructions. Our aim is to construct a PCTL formula $\vp_\M$ such that 
\begin{itemize}
    \item every model of $\vp_\M$ contains a run representing a recurrent computation of~$\M$;
    \item for every recurrent computation of $\M$ there exists a model of $\vp_\M$ representing the computation.
\end{itemize}
This entails the $\Sigma_1^1$ hardness of PCTL satisfiability. 

In the following sections, we indicate how to construct $\varphi_\M$ for a given non-deterministic \emph{one-counter} machine~$\M$. The restriction to one counter  leads to simpler notation without omitting any crucial ideas. The extension to non-deterministic \emph{two-counter} machines is relatively simple (see Section~\ref{sec-twocounter}).
For the rest of this section, we fix a non-deterministic one-counter machine~$\M$ with $m$ instructions.

\subsection{Representing the configurations of $\M$}
\label{sec-conf-repre}
The first step towards constructing the formula $\vp_\M$ is finding a way of representing configurations of $\M$ by ``PCTL-friendly'' properties of states in Markov chains.

More concretely, a state $s$ representing a configuration $(\ell,c)$ must encode the index $\ell$ and the non-negative integer $c$. The  $\ell$ is encoded by a dedicated atomic proposition $\at_\ell$ valid in~$s$ (for technical reasons, $\at_\ell$ is actually a disjunction of several propositions). The way of encoding $c$ is more elaborate. Intuitively, the only unbounded information associated to $s$ that is ``visible'' to PCTL is the probability of satisfying certain path formulae in~$s$. The value of $c$ is encoded by a \emph{pair} of probabilities $\cv[s] \in (0,1)^2$ of satisfying suitable path formulae $\Phi_s$ and $\Psi_s$ in $s$. As we shall see, the formulae $\Phi_s$ and $\Psi_s$ depend on the state~$s$, but the only information about~$s$ used to determine $\Phi_s$ and $\Psi_s$ is the (in)validity of certain atomic propositions in~$s$. Consequently, there are only \emph{finitely many} PCTL formulae used to encode the counter values. The reason why $c$ cannot be represented just by a \emph{single} probability of satisfying some appropriate path formula is indicated in Section~\ref{sec-sim-step}.

It remains to explain which pair of probabilities encodes a given integer. 
Zero is encoded by a fixed (suitably chosen) pair $\vec{z}$, and if $n$ is encoded by $\vec{v}$, then $n+1$ is encoded by $\INC(\vec{v})$, where $\INC : (0,1)^2 \to (0,1)^2$ is a suitable function defined in Section~\ref{sec-geom}. 

To sum up, a state $s$ encodes a configuration $(\ell,c)$ iff $s \models \at_\ell$ and $\cv[s] =  \INC^{c}(\vec{z})$ (recall that $\INC^{c}$ denotes the $c$-fold composition of $\INC$). 

\subsection{Simulating a computational step of $\M$}
\label{sec-sim-step}

The crucial part of our construction is the simulation of one computational step of~$\M$. Let $\ell \in \{1,\ldots,m\}$ where $\Ins_\ell ~\equiv~ \langle C{=}0\, {?}\ Z \mbox{ : } P \rangle:\ \update$. We show that for every $\ell' \in Z \cup P$, there exists a formula $\Step_{\ell,\ell'}$ satisfying the following property: 

\begin{prop}
\label{prop-step}
For every computational step of the form $(\ell,c) \mapsto (\ell',c')$ and every state $s$ representing $(\ell,c)$ we have the following: If $s \models \Step_{\ell,\ell'}$ then 
\begin{enumerate}
    \item[(i)] there exist $\pi \in \run(s)$ and $j\geq 1$ such that $\pi(j)$ represents the configuration $(\ell',c')$;
    \item[(ii)] for every $\pi \in \run(s)$ and every $j \geq 1$ such that $\pi(j)$ represents a configuration of $\M$, there exists $1 \leq j' \leq j$ such that $\pi(j')$ represents the configuration $(\ell',c')$. 
\end{enumerate}
\end{prop}

The construction of $\Step_{\ell,\ell'}$ is perhaps the most tricky part of our proof. At first glance, the above features~(i) and~(ii) seem \emph{unachievable} because of the following fundamental obstacle. 

Suppose that $s$ represents $(\ell,c)$ and $s \models \Step_{\ell,\ell'}$. For simplicity, let us assume that the underlying Markov chain of~$s$ has been unfolded into an infinite tree. According to~(i), there exist $\pi \in \run(s)$ and $j\geq 1$ such that $\pi(j)$ represents the configuration $(\ell',c')$. For brevity, let us denote the states $\pi(j{-}1)$ and $\pi(j)$ by $u$ and $t$, respectively.  Since $t$ represents $(\ell',c')$, we have that $\cv[t] = \INC^{c'}(\vec{z})$ (see Section~\ref{sec-conf-repre}). Due to Property~\ref{prop-step},~(i) and~(ii) must hold in an arbitrary modification $\M'$ of $\M$ such that $s$ still represents $(\ell,c)$ and satisfies $\Step_{\ell,\ell'}$ in $\M'$.  Now consider a modification where the immediate successor $t$ of $u$ is replaced with $t_1,\ldots,t_n$ so that
\begin{enumerate}
    \item[(a)] every $t_i$ satisfies the same set of atomic propositions as $t$,
    \item[(b)] $\Prob(u,t) = \sum_{i=1}^n \Prob(u,t_i)$,
    \item[(c)] $\Prob(u,t) \cdot \m(\Phi,t) = \sum_{i=1}^{n} \Prob(u,t_i) \cdot \m(\Phi,t_i)$~~~ for every path subformula $\Phi$ of $\Step_{\ell,\ell'}$. Here, $\m(\Phi,q)$ is the probability of all runs initiated in $q$ satisfying $\Phi$.
\end{enumerate}
Since this modification does not change the value of $\m(\Phi,s)$ for any path subformula $\Phi$ of $\Step_{\ell,\ell'}$, we have that  $s \models \Step_{\ell,\ell'}$ after the modification. Note that~(c) implies $\Prob(u,t) \cdot \cv[t] = \sum_{i=1}^{n} \Prob(u,t_i) \cdot \cv[t_i]$. However, if we choose $t_1,\ldots,t_n$ so that every $\cv[t_i]$ is \emph{different} from $\cv[t]$, we may spoil both~(i) and~(ii).

Since spoiling the envisioned functionality of $\Step_{\ell,\ell'}$ is so simple, one is even tempted to conclude the \emph{non-existence} of $\Step_{\ell,\ell'}$. Indeed, no suitable $\Step_{\ell,\ell'}$ can exist without some extra restriction preventing the use of \emph{arbitrary} $t_1,\ldots,t_n$ satisfying (a)--(c). We implement such a restriction with the help of basic convex geometry. First, observe that the condition~(c) implies that $\cv[t]$ is a \emph{convex combination} of $\cv[t_1],\ldots,\cv[t_n]$. We design the function~$\INC$
so that the probability pairs representing non-negative integers are \emph{vertices} (i.e., faces of dimension~$0$) of a certain convex polytope. Furthermore, we identify a certain closure property preventing the use of $t_i$ such that $\cv[t_i]$ is \emph{outside} the polytope. Thus, whenever $\cv[t]$ encodes a non-negative integer and $\cv[t]$ is a positive convex combination of $\cv[t_1],\ldots,\cv[t_n]$, we can conclude that every $\cv[t_i]$ is inside the polytope, and consequently every $\cv[t_i]$ is \emph{equal} to $\cv[t]$ because $\cv[t]$ is a vertex of the polytope. Note that this would not be achievable if $\cv[t]$ was a one-dimensional vector, i.e., if the counter values were encoded by the probability of satisfying a single path formula.

The details are given in Section~\ref{sec-geom}. The design of $\INC$ is also influenced by the need of ``implementing'' $\INC$ and its inverse $\DEC$ in PCTL. The implementation requires additional non-trivial tricks described in Section~\ref{sec-formula}. 

\subsection{Constructing the formula $\varphi_\M$}

The formula $\varphi_{\M}$ expressing the existence of a recurrent computation of $\M$ looks as follows:
\begin{equation*}\textstyle
    \varphi_{\M} ~~\equiv~~ \Struct ~\wedge~ \Init ~\wedge~ 
       \opg_{=1} (\bigwedge_{\ell =1}^m (\at_\ell \Rightarrow \Simulate_\ell)) 
       ~\wedge~ \Rec
\end{equation*}
where
\begin{eqnarray*}
    \Rec & \equiv & \textstyle\opg_{=1}\big(\bigwedge_{\ell=1}^m (\at_\ell \Rightarrow \opf_{>0} (\bigvee_{\ell' \in \tau}\at_{\ell'}))\big)\\
    \Simulate_\ell & \equiv & \textstyle(\Zero \Rightarrow \bigvee_{\ell' \in Z^\ell} \Step_{\ell,\ell'}) ~~\wedge~~
    (\neg\Zero \Rightarrow \bigvee_{\ell' \in P^\ell} \Step_{\ell,\ell'})
\end{eqnarray*}

The subformula $\Struct$ enforces certain ``structural properties'' of the model. $\Init$ says that the state satisfying $\varphi_{\M}$ represents the initial configuration $(1,0)$. The subformula $\Simulate_\ell$
says that if the counter $C$ currently holds zero/positive value, then $\Step_{\ell,\ell'}$ holds for some index $\ell'$ of $Z^{\ell}$ or $P^{\ell}$, respectively. $\Rec$ enforces the existence of a run representing a $\tau$-recurrent computation of~$\M$. See Section~\ref{sec-formula} for more details.

\section{Representing Non-Negative Integers by Points in $(0,1)^2$}
\label{sec-geom}

In this section, we show how to represent a non-negative counter value by a pair of quantities, and we design functions modeling the increment and decrement operations on the counter. Our starting point are the results invented for \emph{finite} Markov chains in \cite{ChK:PCTL-finite-sat-LICS}. Our functions $\INC$ and $\DEC$ are modifications of the functions $\sigma$ and $\tau$ of \cite{ChK:PCTL-finite-sat-LICS}, and the set of points determined by $\INC$ 
has a similar structure as the set of points determined by~$\sigma$ (see Fig.~\ref{fig-area}~left). The crucial novel ingredient is Theorem~\ref{thm-area}, which is applicable to \emph{arbitrary} Markov chains, including infinite-state ones. 

As we already indicated in Section~\ref{sec-simulate}, we aim to define suitable $\vec{z} \in (0,1)^2$ and  $\INC : (0,1)^2 \to (0,1)^2$ such that every $n \in \N$  is encoded by $\INC^n(\vec{z})$, where $\INC^n$ is the $n$-fold composition of $\INC$. The points $\vec{z}$, $\INC(\vec{z})$, $\INC^2(\vec{z}),\ldots$ should be vertices of a certain convex polytope, and both $\INC$ and its inverse $\DEC$ should be ``expressible in PCTL''. Rougly speaking, the second condition means that if we have a state $s$ representing a configuration $(\ell,c_1,c_2)$, then there exists a PCTL formula enforcing the existence of a successor of $s$ encoding the values obtained from $c_1,c_2$ by performing the instructions $\update_1^\ell$ and $\update_2^\ell$. After many unsuccessful trials, we have identified a function satisfying these requirements. We put
\[
   \INC(\vec{v}) \ = \ \left(\frac{\lambda}{1-\vec{v}_1}, \frac{\vec{v}_2 \cdot \lambda}{1-\vec{v}_1}\right)
\]
where $\lambda \in (0,\frac{1}{4})$ is a fixed constant. In Section~\ref{sec-formula}, we set $\lambda$ to a specific value. All observations presented in this section are valid for an \emph{arbitrary} $\lambda \in (0,\frac{1}{4})$.

At first glance, $\INC$ is \emph{not} a good choice because it is not even a function $(0,1)^2 \to (0,1)^2$. However, there exists a \emph{subinterval} $I_\lambda \subseteq (0,1)$ such that $\INC : I_\lambda \times (0,1) \to I_\lambda \times (0,1)$. Furthermore, let
\[
   \DEC(\vec{v}) \ = \ \left(\frac{\vec{v}_1 - \lambda}{\vec{v}_1}, \frac{\vec{v}_2}{\vec{v}_1}\right)\,.
\]
The properties of $\INC$ and $\DEC$ are summarized in Lemma~\ref{lem-tausigma}. 

The \emph{slope} of a line $\Line$ in $\R^2$, denoted by $\slope(\Line)$, is defined in the standard way. For all $\vec{v},\vec{u} \in \R^2$ where $\vec{v}_1 \neq \vec{u}_1$, we use $\slope(\vec{v},\vec{u})$ to denote the slope of the line containing $\vec{v},\vec{u}$.

\begin{restatable}{lemma}{incfunction}
    \label{lem-tausigma}
    Let $\lambda \in (0,\frac{1}{4})$ and
    \(
   I_\lambda \ = \ \left(\frac{1-\sqrt{1-4\lambda}}{2},  \frac{1+\sqrt{1-4\lambda}}{2}\right)
    \). For every $\vec{v} \in I_\lambda \times [0,1]$ we have the following: 
    \begin{enumerate}
        \item[(a)] $\INC(\vec{v}) \in I_\lambda \times [0,1]$;
        \item[(b)] $\DEC(\INC(\vec{v})) = \vec{v}$;
        \item[(c)] $\INC(\vec{v})_1 < \vec{v}_1$ and $\INC(\vec{v})_2 \leq \vec{v}_2$; if $\vec{v}_2 > 0$, then $\INC(\vec{v})_2 < \vec{v}_2$; 
        \item[(d)] let $\vec{u} = (\INC^2(\vec{v})_1,0)$; then $\slope(\vec{u},\INC(\vec{v})) = \slope(\INC(\vec{v}),\vec{v})$;
        \item[(e)] let $\vec{u} = (\INC(\vec{v})_1,y)$ where $0 \leq y <  \INC(\vec{v})_2$. Then $\slope(\vec{u},\DEC(\vec{u})) < \slope(\INC(\vec{v}),\vec{v})$;
        \item[(f)] if $\vec{u}$ belongs to the line segment between $\INC^2(\vec{v})$ and $\INC(\vec{v})$, then $\DEC(\vec{u})$ belongs to the line segment between $\INC(\vec{v})$ and $\vec{v}$;
        \item[(g)] $\lim_{n \to \infty} \INC^n(\vec{v})_1 =  \frac{1-\sqrt{1-4\lambda}}{2}$.
    \end{enumerate}
\end{restatable}

Observe that by choosing a sufficiently small $\lambda>0$, the extreme points of $I_\lambda$ can be pushed to values arbitrarily close to $0$ and $1$. Let us fix some $\vec{z} \in I_\lambda \times (0,1)$. The structure of $\vec{z}, \INC(\vec{z}), \INC^2(\vec{z}),\ldots$ is shown in Fig.~\ref{fig-area}~(left), where the dotted lines illustrate the property of Lemma~\ref{lem-tausigma}~(d). Furthermore, Fig.~\ref{fig-area}~(left) also shows the convex polytope $\Area(\vec{z})$ whose boundaries are the vertical lines going through the points $(\frac{1-\sqrt{1-4\lambda}}{2},0)$ and~$\vec{z}$, the horizontal line going through the point $(0,1)$, and all line segments between the consecutive points $\INC^i(\vec{z})$ and $\INC^{i+1}(\vec{z})$. Note that $\Area(\vec{z})$ is indeed convex by Lemma~\ref{lem-tausigma}. Furthermore, each $\INC^i(\vec{z})$ is a \emph{vertex} (i.e., a \emph{face} of dimension~$0$) of $\Area(\vec{z})$.
In our proof of Theorem~\ref{thm-area}, we also need the following~lemma:
\begin{restatable}{lemma}{outlineseg}
    \label{lem-outlineseg}
    Let $\lambda \in (0,\frac{1}{4})$ and $\vec{z} \in I_\lambda \times (0,1)$.
    For all $\vec{v} \in I_\lambda \times [0,1] \smallsetminus \Area(\vec{z})$ and $n \in \N$ such that $\INC^{n+1}(\vec{z})_1 \leq \vec{v}_1 \leq \INC^{n}(\vec{z})_1$, there exists $\vec{u}$ such that  
    $\vec{u}_1 = \INC^{n+1}(\vec{z})_1$, $0 \leq \vec{u}_2 < \INC^{n+1}(\vec{z})_2$, and $\vec{v}$ belongs to the line segment between $\vec{u}$ and $\DEC(\vec{u})$.
\end{restatable}

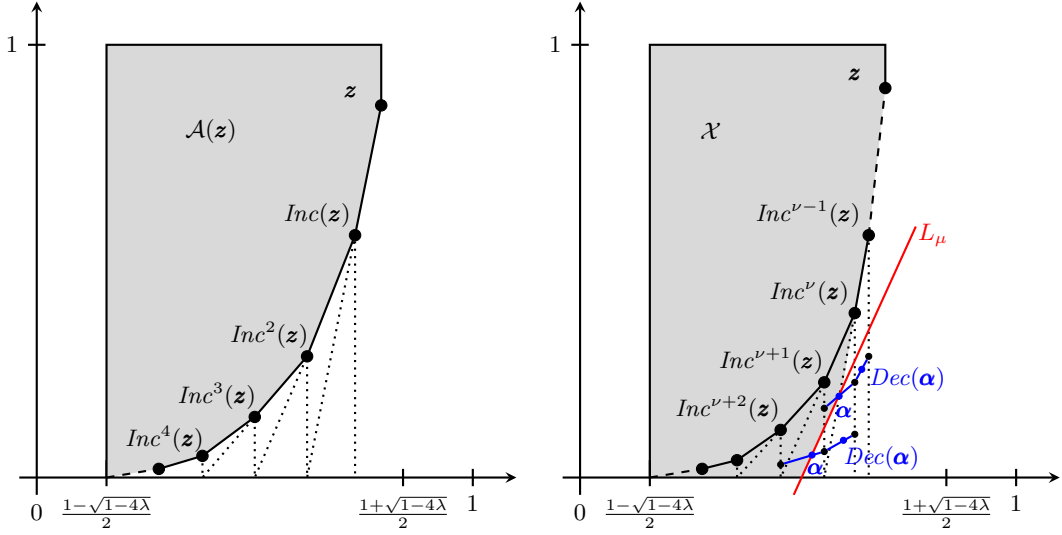
\begin{figure}[t]\centering
    \begin{tikzpicture}[x=1.35cm, y=1.35cm,font=\small,scale=0.85]   
    \path[thick,-stealth, at start] 
        (-0.2,0) edge  (5.5,0)
        (0,-0.2) edge node[below] {$0$} (0,5.5);
        \path[thick,-, at start]
        (-.1,5)  edge node[left] {$1$} (.1,5);
        \path[thick,-, at start]
            (5,-.1)  edge node[below] {$1$} (5,0.1);
        \path[thick, at start] 
            (.8,-0.1) edge node[below] {$\frac{1-\sqrt{1-4\lambda}}{2}$} (0.8,.2)
            (4.2,-0.1) edge node[below] {$\frac{1+\sqrt{1-4\lambda}}{2}$} (4.2,.1);  
    \tikzset{ver/.style={draw, circle, inner sep=0pt, minimum size=.15cm, fill=black}}
    \coordinate (a) at (1.4,.1);
    \coordinate (b) at (1.9,.25);
    \coordinate (b0) at (1.9,0);
    \coordinate (c)  at (2.5,.7);
    \coordinate (c0) at (2.5,0);
    \coordinate (d) at (3.1,1.4);
    \coordinate (d0) at (3.1,0);
    \coordinate (e) at (3.65,2.8);
    \coordinate (e0) at (3.65,0);
    \coordinate (f) at (3.95,4.3);
    \coordinate (s) at (.8,0);
    \coordinate (t) at (3.95,5);
    \coordinate (tt) at (0.8,5);
    \draw [fill=gray!30,draw opacity=0] plot coordinates {(s) (a) (b) (c) (d) (e) (f) (t) (tt)};
    \draw[thick, dotted] (b) -- (b0) -- (c)-- (c0) -- (d)-- (d0) -- (e)-- (e0); 
    \node[ver] (A) at (a) {};
    \node[ver, label={above, xshift=-5mm, yshift=-1.5mm: $\INC^4(\vec{z})$}] (B) at (b) {};
    \node[ver, label={above, xshift=-5mm, yshift=-1mm: $\INC^3(\vec{z})$}] (C) at (c) {};
    \node[ver, label={above, xshift=-5mm, yshift=-1mm: $\INC^2(\vec{z})$}] (D) at (d) {};
    \node[ver, label={above, xshift=-5mm, yshift=-1mm: $\INC(\vec{z})$}] (E) at (e) {};    
    \node[ver, label={above, xshift=-4mm, yshift=-1mm: $\vec{z}$}] (F) at (f) {};
    \draw[thick] (a) -- (b) -- (c) -- (d) -- (e) -- (f) -- (t) -- (tt) -- (s);
    \draw[thick,dashed] (s) -- (a);
    \node (T) at (2,4) {$\Area(\vec{z})$};
\end{tikzpicture}\hspace*{.2cm}
\begin{tikzpicture}[x=1.35cm, y=1.35cm,font=\small,scale=0.85]
    \tikzset{ver/.style={draw, circle, inner sep=0pt, minimum size=.15cm, fill=black}}    
    \tikzset{vr/.style={draw, circle, inner sep=0pt, minimum size=.08cm, fill=black}}    
    \coordinate (a) at (1.4,.1);
    \coordinate (b) at (1.8,.2);
    \coordinate (b0) at (1.8,0);
    \coordinate (c)  at (2.3,.55);
    \coordinate (c0) at (2.3,0);
    \coordinate (d) at (2.8,1.1);
    \coordinate (d0) at (2.8,0);
    \coordinate (e) at (3.15,1.9);
    \coordinate (e0) at (3.15,0);
    \coordinate (ee)  at (3.31,2.8);
    \coordinate (ee0) at (3.31,0);
    \coordinate (f) at (3.5,4.5);
    \coordinate (f0) at (3.5,0);
    \coordinate (s) at (.8,0);
    \coordinate (t) at (3.5,5);
    \coordinate (tt) at (0.8,5);
    \draw [fill=gray!30,draw opacity=0] plot coordinates {(s) ([shift={(0.24,0)}]c0) ([shift={(0.22,0)}]d) (e) (ee) (f) (t) (tt)};
    \path[thick,-stealth, at start] 
    (-0.2,0) edge  (5.5,0)
    (0,-0.2) edge node[below] {$0$} (0,5.5);
    \path[thick,-, at start]
        (5,-.1)  edge node[below] {$1$} (5,0.1);
    \path[thick,-, at start]
        (-.1,5)  edge node[left] {$1$} (.1,5);
    \path[thick, at start] 
        (.8,-0.1) edge node[below] {$\frac{1-\sqrt{1-4\lambda}}{2}$} (0.8,.2)
        (4.2,-0.1) edge node[below] {$\frac{1+\sqrt{1-4\lambda}}{2}$} (4.2,.1);  
    \draw[thick, dotted] (b) -- (b0) -- (c) -- (c0) -- (d) -- (d0) -- (e) -- (e0);
    \draw[thick, dotted] (ee) -- (ee0);
    \draw[thick] (f) -- (t) -- (tt) -- (s);
    \node[ver] (A) at (a) {};
    \node[ver] (B) at (b) {};
    \node[ver, label={above, xshift=-7mm, yshift=-1mm: $\INC^{\nu+2}(\vec{z})$}] (C) at (c) {};
    \node[ver, label={above, xshift=-7mm, yshift=-1mm: $\INC^{\nu+1}(\vec{z})$}] (D) at (d) {};
    \node[ver, label={above, xshift=-6mm, yshift=-1mm: $\INC^{\nu}(\vec{z})$}] (E) at (e) {};
    \node[ver, label={above, xshift=-8mm, yshift=-1mm: $\INC^{\nu-1}(\vec{z})$}] (E) at (ee) {};
    \node[ver, label={above, xshift=-4mm, yshift=-1mm: $\vec{z}$}] (F) at (f) {};
    \draw[thick] (a) -- (b) -- (c) -- (d) -- (e) -- (ee);
    \draw[thick,dashed] (s) -- (a);
    \draw[thick,dashed] (ee) -- (f);
    \draw[thick,color=red] ([shift={(0.15,-0.2)}]c0) -- ([shift={(0.7,1.0)}] e);
    \node[thick,color=red] at ([shift={(.9,.9)}] e) {$L_\mu$};
    \node[vr] (r1) at ([shift={(0,-.4)}]c) {};
    \node[vr] (r2) at ([shift={(0,-.8)}]d) {};
    \node[vr] (r3) at ([shift={(0,-1.4)}]e) {};
    \draw[thick,color=blue] (r1) -- (r2) -- (r3);   
    \node[vr,color=blue,label={below, color=blue, xshift=.5mm, yshift=.6mm: $\vec{\alpha}$}] (rho1) at ([shift={(.36,-.29)}]c) {};
    \node[vr,color=blue,label={below, color=blue, xshift=5mm, yshift=1mm: $\DEC(\vec{\alpha})$}] (trho1) at ([shift={(.22,-.67)}]d) {};
    \node[vr] (rr1) at ([shift={(0,-.3)}]d) {};
    \node[vr] (rr2) at ([shift={(0,-.8)}]e) {};
    \node[vr] (rr3) at ([shift={(0,-1.4)}]ee) {};
    \draw[thick,color=blue] (rr1) -- (rr2) -- (rr3);   
    \node[vr,color=blue,label={above, color=blue, xshift=.7mm, yshift=-4.5mm: $\vec{\alpha}$}] (rho2) at 
         ([shift={(.17,-.16)}]d) {};
    \node[vr,color=blue,label={below, color=blue, xshift=6mm, yshift=2mm: $\DEC(\vec{\alpha})$}] (trho2) at 
         ([shift={(.08,-.65)}]e) {}; 
    \node (T) at (1.5,4) {$\X$};
\end{tikzpicture}
\caption{The structure of $\vec{z}, \INC(\vec{z}), \INC^2(\vec{z}),\ldots$ and $\Area(\vec{z})$ (left); the construction proving $\Out =\emptyset$ (right). }
\label{fig-area}
\end{figure}

The next theorem is a crucial tool enabling the construction of the formula $\Step_{\ell,\ell'}$ (see Section~\ref{sec-sim-step}). 

\begin{theorem}
    \label{thm-area}
    Let $\lambda \in (0,\frac{1}{4})$ and $\vec{z} \in I_\lambda \times (0,1)$.
    Furthermore, let $T$ be a finite or countably infinite subset of $(\frac{1-\sqrt{1-4\lambda}}{2},\vec{z}_1) \times (0,1)$ such that for every $\vec{v} \in T \smallsetminus \Area(\vec{z})$, the pair $\DEC(\vec{v})$ is a convex combination of the elements of~$T$. Then $T \subseteq \Area(\vec{z})$.
\end{theorem}
\begin{proof}
   We start by introducing some notation. For every line $\Line$ such that $\slope(\Line) \in (0,\infty)$, we use $H(\Line)$ to denote the closed half-space in $\R^2$ above $\Line$. Furthermore, for every $n\in \N$, we use $\Line_n(\vec{z})$ to denote the line containing the points $\INC^n(\vec{z})$ and $\INC^{n+1}(\vec{z})$. Note that $\Area(\vec{z}) \subseteq H(\Line_n(\vec{z}))$ for every $n \in \N$.
   
   Now let $T$ be a set satisfying the assumptions of our lemma, and  let $\Out = T \smallsetminus \Area(\vec{z})$. Note that if $\vec{v} \in \Out$, then there exists $n \in \N$ such that $\INC^{n+1}(\vec{z})_1 \leq \vec{v}_1 \leq \INC^{n}(\vec{z})_1$ by Lemma~\ref{lem-tausigma}~(g). Observe that if $\vec{v} \in H(\Line_n(\vec{z}))$ then $\vec{v} \in \Area(\vec{z})$, which is a contradiction. Hence, $\vec{v} \not\in H(\Line_n(\vec{z}))$.
   
   Our aim is to show that $\Out = \emptyset$. Suppose the converse. Due to the above observation, there exists $n \in \N$ such that $\Out \smallsetminus H(\Line_n(\vec{z})) \neq \emptyset$.
   Let $\nu$ be the \emph{least} \mbox{$n \in \N$} such that $\Out \smallsetminus H(\Line_n(\vec{z})) \neq \emptyset$. For every \mbox{$0 \leq y \leq \INC^{\nu+1}(\vec{z})_2$}, let $L_y$ be the line with the same slope as $\Line_{\nu}(\vec{z})$ containing the point $(\INC^{\nu+1}(\vec{z})_1, y)$. 

   Suppose that there exists $\vec{v} \in \Out$ such that $\vec{v} \not\in H(L_0)$. Then
   \begin{itemize}
    \item if $\nu = 0$, we have that $\DEC(\vec{v})_1 > \vec{z}_1$, hence $\DEC(\vec{v})$ is not a convex combination of the elements of $T$, and we have a contradiction; 
    \item if $\nu > 0$, we obtain $\vec{v} \not\in H(\Line_{\nu-1}(\vec{z}))$, which contradicts the minimality of $\nu$.
   \end{itemize}
    Now suppose $\Out \subseteq H(L_0)$. Let
   $\mu$ be the supremum of all \mbox{$y \leq \INC^{\nu+1}(\vec{z})_2$} such that
    $\Out \subseteq H(L_y)$. Clearly, $0 \leq \mu < \INC^{\nu+1}(\vec{z})_2$. 
    Let 
    \[
          \X \quad = \quad \left[\textstyle\frac{1-\sqrt{1-4\lambda}}{2},\vec{z}_1\right] \times [0,1] 
           \ \ \cap \ \ H(L_\mu) \ \ \cap \ \ 
           \bigcap_{n=0}^{\nu-1} H(\Line_n(\vec{z}))
    \]
    Note $\X$ is a convex set and $T \subseteq \X$.  We show that there exists $\vec{v} \in \Out$ such that $\DEC(\vec{v}) \not\in \X$. By the 
    assumptions of our lemma, $\DEC(\vec{v})$ is a convex combination of the elements of~$T$. Since $T \subseteq \X$ and $\X$ is a convex set, this is impossible, and we have the desired contradiction.
    
    The existence of $\vec{v} \in \Out$ such that $\DEC(\vec{v}) \not\in \X$ is proven as follows. By the definition of $\mu$, there is an infinite sequence $\vec{u}^1,\vec{u}^2,\ldots$ such that $\vec{u}^i \in \Out$ for all $i \geq 1$ and the distance of $\vec{u}^i$ from $L_\mu$ approaches $0$ as $i \to \infty$.
    This sequence must contain an infinite subsequence converging to some point $\vec{\alpha} \in L_\mu \cap \X$. 
        
    Since $\DEC$ is continuous, it suffices to show that $\DEC(\vec{\alpha}) \not\in \X$ for all $\vec{\alpha} \in L_\mu \cap \X$. Let us fix some $\vec{\alpha} \in L_\mu \cap \X$. It follows directly from the definition of $L_\mu$ that $\INC^{\nu+2}(\vec{z})_1 < \vec{\alpha}_1 < \INC^{\nu}(\vec{z})_1$ (see Fig. \ref{fig-area}~right). Hence, we have that either $\INC^{\nu+2}(\vec{z})_1 < \vec{\alpha}_1 \leq \INC^{\nu+1}(\vec{z})_1$ or $\INC^{\nu+1}(\vec{z})_1 < \vec{\alpha}_1 < \INC^{\nu}(\vec{z})_1$. Let us first consider the case when $\INC^{\nu+2}(\vec{z})_1 < \vec{\alpha}_1 \leq \INC^{\nu+1}(\vec{z})_1$. By Lemma~\ref{lem-outlineseg}, there exists $\vec{u}$ such that $\vec{u}_1 = \INC^{\nu+2}(\vec{z})_1$, $0 \leq \vec{u}_2 < \INC^{\nu+2}(\vec{z})_2$, and $\vec{\alpha}$ belongs to the line segment between $\vec{u}$ and $\DEC(\vec{u})$. By Lemma~\ref{lem-tausigma}~(f), the point $\DEC(\vec{\alpha})$ then belongs to the line segment between $\DEC(\vec{u})$ and $\DEC^2(\vec{u})$. However, this line segment is disjoint with $\X$ because its slope is strictly less than the slope of $L_\mu$ (see Lemma~\ref{lem-tausigma}). A similar argument applies in the case when $\INC^{\nu+1}(\vec{z})_1 < \vec{\alpha}_1 < \INC^{\nu}(\vec{z})_1$ (see Fig. \ref{fig-area}~right).
\end{proof}

Observe that if $T$ is a set satisfying the assumptions of Theorem~\ref{thm-area}, then for every $n \in \N$ we have that if $\INC^n(\vec{z}) = \sum_{\vec{v} \in T} \kappa(\vec{v}) \cdot \vec{v}$ for some probability distribution $\kappa$ over $T$, then $\vec{v} = \INC^n(\vec{z})$ for all $\vec{v} \in T$ such that $\kappa(\vec{v}) > 0$ (because $T \subseteq \Area(\vec{z})$ and $\INC^n(\vec{z})$ is a vertex of~$\Area(\vec{z})$). When constructing the formula~$\varphi_\M$ (see Section~\ref{sec-formula}), we ensure that for all $k \in \{1,2\}$, the set of all $\cv^k[t]$, where $t$ ranges over a set of ``relevant'' states, satisfies the conditions of Theorem~\ref{thm-area}. Thus, we overcome the obstacle mentioned in Section~\ref{sec-conf-repre}.

\section{Constructing the Formula $\varphi_\M$}
\label{sec-formula}

In this section, we show how to construct the formula $\varphi_\M$ simulating a given non-deterministic one-counter machine~$\M$ with $m$~instructions. The extension to two-counter machines is relatively simple and it is presented in Section~\ref{sec-twocounter}. 

When explaining the meaning of the constructed subformulae of $\varphi_\M$, we often refer to some unspecified model of $\varphi_\M$. Without restrictions, we assume that all states of the considered model are reachable from a state satisfying $\varphi_\M$.
The structure of $\varphi_\M$ has already been presented in Section~\ref{sec-simulate}. Recall that
\begin{equation*}\textstyle
    \varphi_{\M} ~~\equiv~~ \Struct ~\wedge~ \Init ~\wedge~ 
       \opg_{=1} (\bigwedge_{\ell =1}^m (\at_\ell \Rightarrow \Simulate_\ell)) 
       ~\wedge~ \Rec
\end{equation*}
where
\begin{eqnarray*}
    \Rec & \equiv & \textstyle\opg_{=1}\big(\bigwedge_{\ell=1}^m (\at_\ell \Rightarrow \opf_{>0} (\bigvee_{\ell' \in \tau}\at_{\ell'}))\big)\\
    \Simulate_\ell & \equiv & \textstyle(\Zero \Rightarrow \bigvee_{\ell' \in Z^\ell} \Step_{\ell,\ell'}) ~~\wedge~~
    (\neg\Zero \Rightarrow \bigvee_{\ell' \in P^\ell} \Step_{\ell,\ell'})
\end{eqnarray*}
We show how to implement the subformulae of $\varphi_\M$, where $\Step_{\ell,\ell'}$ is the main challenge.

Recall that the constant~$\lambda$ of Section~\ref{sec-geom} can be arbitrarily small. For our purposes, we need to choose $\lambda$, $\vec{z}_1$, $\vec{z}_2$, and $\delta$ (which exists for technical reasons) so that the extreme points of $I_\lambda$ are rational, $\vec{z}_2 < \vec{z}_1 < \delta$, and
$2\lambda + 2\delta + 2\vec{z}_1 + 2\vec{z}_2 < 1$.
For example, we can put $\lambda = \frac{14}{255}$, $I_\lambda =(\frac{1}{15},\frac{14}{15})$, $\vec{z} = (\frac{1}{12},\frac{1}{15})$, and $\delta = \frac{1}{11}$.  Furthermore, let
%
%
\begin{eqnarray*}
    \B & = & \{K\} \cup \{A_i,B_i,C_i,D_i,E_i,R_i \mid 0 \leq i \leq 2\}\\
    \A & = & \B \cup  \{a_{i,\ell},\bar{a}_{i,\ell},b_{i,\ell},c_{i,\ell},d_{i,\ell},r_{i,\ell} \mid 0 \leq i \leq 2, 1 \leq \ell \leq m\}
\end{eqnarray*}
be sets of atomic propositions.
For every $i \in \{0,1,2\}$, we use $S(i)$ to denote $i{+}1 \mbox{ mod } 3$, and
for all $O \subseteq \A$ and $L \subseteq O$, let
\(
    \ex{L}_O \ \equiv \  \bigwedge_{p \in L} p  ~\wedge~ \bigwedge_{p \in O \smallsetminus L} \neg p\,.
\)

\subsection{Defining the basic structure of a model}
\label{sec-structure}
Intuitively, states satisfying $r_{i,\ell}$ are the states where $\varphi_\M$ simulates the instruction $\Ins_\ell$ of $\M$. The role of the index $i$ in $r_{i,\ell}$ is auxiliary. For every $1 \leq \ell \leq m$, we put $\at_\ell \equiv \ex{r_{0,\ell}}_\A \vee \ex{r_{1,\ell}}_\A \vee \ex{r_{2,\ell}}_\A$. As we shall see, a state satisfying $\varphi_\M$ satisfies $\ex{r_{0,1}}_{\A}$.

The formula $\Struct$ defines the basic structure of a model of $\varphi_\M$.  We put 
\begin{eqnarray*}
  \Struct & \equiv & \Succ \ \wedge \ \Mark \ \wedge\ \FLambda
\end{eqnarray*}
where
\begin{eqnarray*}
    \Succ & \equiv & \bigwedge_{\ell=1}^m \bigwedge_{i=0}^2 \opg_{=1} \bigg(  \ex{r_{i,\ell}}_{\A} \ \Rightarrow \ \bigvee_{\ell'=1}^m \ex{r_{i,\ell}}_{\A} \opu_{=1} \rSuc_{i,\ell,\ell'} \bigg)\\ 
       & \wedge & \bigwedge_{i=0}^2 \opg_{=1} \bigg(
       \big( \ex{R_i}_{\B} \ \Rightarrow \ \ex{R_i}_{\B} \opu_{=1} \RSuc_i \big) 
       \wedge
       \big( \ex{R_i,K}_{\B} \ \Rightarrow \ \ex{R_i,K}_{\B} \opu_{=1} \RKSuc_i \big) 
       \bigg)
\end{eqnarray*}
where
{\small
\begin{eqnarray*}
    \rSuc_{i,\ell,\ell'} & \equiv & \ex{r_{S(i),\ell'}}_{\A} \vee \bigvee_{j=0}^2 \left(\ex{a_{i,\ell},R_j}_{\A} \ \vee \ex{\bar{a}_{i,\ell},R_j}_{\A}\ \vee \ \ex{b_{i,\ell}, R_j}_{\A} \ \vee \ \ex{c_{i,\ell},R_j}_{\A} \ \vee \ \ex{d_{i,\ell},R_j}_{\A}\right)  \\
    \RSuc_i & \equiv & \ex{A_i}_{\B} \ \vee \ \ex{B_i}_{\B} \ \vee \ \ex{C_i}_{\B} \ \vee \ \ex{D_i}_{\B} \ \vee \ \ex{E_i}_{\B} \ \vee \ \ex{R_{S(i)},K}_{\B}\\
    \RKSuc_i & \equiv & \ex{A_i,K}_{\B} \ \vee \ \ex{B_i,K}_{\B} \ \vee \ \ex{C_i,K}_{\B} \ \vee \ \ex{D_i,K}_{\B}  \ \vee \ \ex{R_{S(i)},K}_{\B}
\end{eqnarray*}}%
Let $O \subseteq \A$ and $L \subseteq O$. We say that $L$ is an \emph{$O$ marker} if whenever $s \models \ex{L}_O$, then $s' \models \ex{L}_O$ for every~$s'$ reachable from~$s$. Note that this property can be encoded by the PCTL formula $\ex{L}_O \Rightarrow \opg_{=1} \ex{L}_O$. The formula $\Mark$ says that
\begin{itemize}
    \item the sets $\{a_{i,\ell}\}$, $\{\bar{a}_{i,\ell}\}$, $\{b_{i,\ell}\}$, $\{c_{i,\ell}\}$, and $\{d_{i,\ell}\}$ are $\A \smallsetminus \B$ markers,
    \item the sets $\{A_i\}$, $\{A_i,K\}$, $\{B_i\}$, $\{B_i,K\}$, $\{C_i\}$, $\{C_i,K\}$, $\{D_i\}$, $\{D_i,K\}$, and $\{E_i\}$ are $\B$ markers.
\end{itemize} 
Finally, $\FLambda \equiv \bigwedge_{i=0}^2 (\ex{R_i}_{\B} \Rightarrow \opg_{=\lambda}(R_i {\vee} E_i)$.
The purpose of $\FLambda$ is clarified when constructing the formula $\Step_{\ell,\ell'}$. 

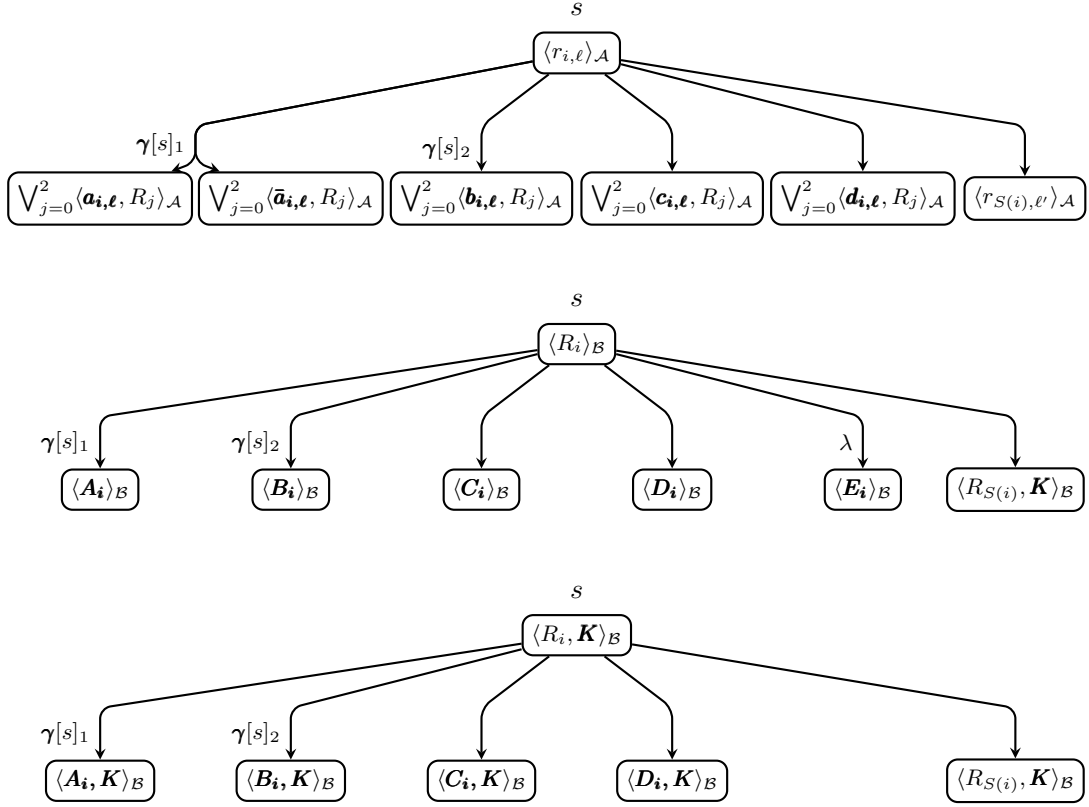
\begin{figure}[t]\centering
\begin{tikzpicture}[x=2.6cm, y=2cm,font=\footnotesize,scale=.965]
    \node[state,align=center] (s) at (0,0) 
            {$\ex{r_{i,\ell}}_\A$};
    \node at (0,.3) {{\large$s$}};
    \node[state] (a)  at (-2.5,-1) {$\bigvee_{j=0}^2\ex{\pmb{a_{i,\ell}},R_j}_\A$};
    \node[state] (ab) at (-1.5,-1) {$\bigvee_{j=0}^2\ex{\pmb{\bar{a}_{i,\ell}},R_j}_\A$};
    \node[state] (b)  at (-0.5,-1) {$\bigvee_{j=0}^2\ex{\pmb{b_{i,\ell}},R_j}_\A$};
    \node[state] (c)  at (0.5,-1)  {$\bigvee_{j=0}^2\ex{\pmb{c_{i,\ell}},R_j}_\A$};
    \node[state] (d)  at (1.5,-1)  {$\bigvee_{j=0}^2\ex{\pmb{d_{i,\ell}},R_j}_\A$};
    \node[state] (r)  at (2.35,-1)  {$\ex{r_{S(i),\ell'}}_\A$};
    \draw[tran] (s)  --  (-2,-.5) -- node[left]{$\cv[s]_1$} (-2,-.75) -- (a);
    \draw[tran] (s)  --  (-2,-.5) -- (-2,-.75) -- (ab);
    \draw[tran] (s)  --  ($(b) +(0,0.5)$) -- node[left]{$\cv[s]_2$} (b);
    \draw[tran] (s)  --  ($(c) +(0,0.5)$) -- (c);
    \draw[tran] (s)  --  ($(d) +(0,0.5)$) -- (d);
    \draw[tran] (s)  --  ($(r) +(0,0.5)$) -- (r);
\tikzset{shift={(0,-2)}};
   \node[state,align=center] (s) at (0,0) {$\ex{R_{i}}_\B$};
    \node at (0,.3) {{\large$s$}};
    \node[state] (a)  at (-2.5,-1) {$\ex{\pmb{A_{i}}}_\B$};
    \node[state] (b)  at (-1.5,-1) {$\ex{\pmb{B_{i}}}_\B$};
    \node[state] (c)  at (-0.5,-1) {$\ex{\pmb{C_{i}}}_\B$};
    \node[state] (d)  at (0.5,-1)  {$\ex{\pmb{D_{i}}}_\B$};
    \node[state] (e)  at (1.5,-1)  {$\ex{\pmb{E_{i}}}_\B$};
    \node[state] (r)  at (2.3,-1)  {$\ex{R_{S(i)},\pmb{K}}_\B$};
    \draw[tran] (s)  --  ($(a) +(0,0.5)$) -- node[left]{$\cv[s]_1$} (a);
    \draw[tran] (s)  --  ($(b) +(0,0.5)$) -- node[left]{$\cv[s]_2$} (b);
    \draw[tran] (s)  --  ($(c) +(0,0.5)$) -- (c);
    \draw[tran] (s)  --  ($(d) +(0,0.5)$) -- (d);
    \draw[tran] (s)  --  ($(e) +(0,0.5)$) -- node[left]{$\lambda$} (e);
    \draw[tran] (s)  --  ($(r) +(0,0.5)$) -- (r);
\tikzset{shift={(0,-2)}};
    \node[state,align=center] (s) at (0,0) {$\ex{R_{i},\pmb{K}}_\B$};
    \node at (0,.3) {{\large$s$}};
    \node[state] (a)  at (-2.5,-1) {$\ex{\pmb{A_{i},K}}_\B$};
    \node[state] (b)  at (-1.5,-1) {$\ex{\pmb{B_{i},K}}_\B$};
    \node[state] (c)  at (-0.5,-1) {$\ex{\pmb{C_{i},K}}_\B$};
    \node[state] (d)  at (0.5,-1)  {$\ex{\pmb{D_{i},K}}_\B$};
    \node[state] (r)  at (2.3,-1)  {$\ex{R_{S(i)},\pmb{K}}_\B$};
    \draw[tran] (s)  --  ($(a) +(0,0.5)$) -- node[left]{$\cv[s]_1$} (a);
    \draw[tran] (s)  --  ($(b) +(0,0.5)$) -- node[left]{$\cv[s]_2$} (b);
    \draw[tran] (s)  --  ($(c) +(0,0.5)$) -- (c);
    \draw[tran] (s)  --  ($(d) +(0,0.5)$) -- (d);
    \draw[tran] (s)  --  ($(r) +(0,0.5)$) -- (r);   
\end{tikzpicture}
\caption{Illustrating the meaning of $\Struct$.}
\label{fig-struct-mb}
\end{figure}
    
The meaning of $\Struct$ is illustrated in Fig.~\ref{fig-struct-mb}. The first subformula of $\Succ$ guarantees that if $s \models \ex{r_{i,\ell}}_{\A}$, then for almost all $\pi \in \run(s)$ there is $n \in \N$ such that
\begin{itemize}
    \item $\pi(n') \models \ex{r_{i,\ell}}_\A$ for every $0 \leq n' <n$;
    \item $\pi(n)$ satisfies either $\bigvee_{j=0}^2\ex{a_{i,\ell},R_j}_{\A}$, or 
$\bigvee_{j=0}^2\ex{\bar{a}_{i,\ell},R_j}_{\A}$, or
$\bigvee_{j=0}^2\ex{b_{i,\ell}, R_j}_{\A}$, or 
$\bigvee_{j=0}^2\ex{c_{i,\ell},R_j}_{\A}$, or 
$\bigvee_{j=0}^2\ex{d_{i,\ell},R_j}_{\A}$, or $\ex{r_{S(i),\ell'}}_{\A}$.
\end{itemize}
This is shown in Fig.~\ref{fig-struct-mb}, together with the properties of states satisfying $\ex{R_i}_\B$ and $\ex{R_i,K}_\B$ enforced by the other subformulae of $\Succ$. The markers are typeset in boldface, and the probabilities constituting $\cv[s]_1$ and $\cv[s]_2$ are also shown. The meaning of $\FLambda$ is apparent in Fig.~\ref{fig-struct-mb}~(middle).

The three graphs of Fig.~\ref{fig-struct-mb} describe the structure of every model of $\varphi_\M$. As we shall see, $s \models \varphi_\M$ implies $s \models \ex{r_{0,1}}_\A$. Almost every run initiated in $s$ eventually visits a state $s'$ satisfying one of the six formulae shown in Fig.~\ref{fig-struct-mb}~(top). For example, suppose that $s' \models \ex{b_{i,\ell}, R_j}_\B$. Since $b_{i,\ell}$ is a marker, this proposition is valid in all successors of~$s'$. Furthermore, $s' \models \ex{R_j}_\B$, and the structure of states reachable from $s'$ is described in Fig.~\ref{fig-struct-mb}~(middle), i.e., almost every run initiated in $s'$ eventually visits a state $s''$ satisfying one of the six formulae shown in Fig.~\ref{fig-struct-mb}~(middle). If $s''$ satisfies one the first five formulae, the states reachable from $s''$ are no longer interesting (all of them satisfy the associated $\B$ marker). If $s'' \models \ex{R_{S(j)},K}_\B$, the structure of states reachable from $s''$ is shown in Fig.~\ref{fig-struct-mb}~(bottom).

The most interesting runs initiated in a state $s$ such that $s \models\varphi_\M$ are the runs visiting infinitely many states satisfying a formula of the form $\ex{r_{i},\ell}$ (in Fig.~\ref{fig-struct-mb}~(top), keep following the rightmost branch). These runs are responsible for modelling the computations of~$\M$.

\subsection{Encoding the counter value}
Now we define the path formulae $\Phi_s$ and $\Psi_s$ encoding the current value of the counter $C$ (see Section~\ref{sec-conf-repre}). Let $\calR = \{r_{i,\ell},R_i \mid 0 \leq i \leq 2, 1 \leq \ell \leq m\}$. A state is \emph{relevant} if it satisfies some proposition of $\calR$ (note that according to $\Struct$, every relevant state satisfies \emph{exactly one} proposition of $\calR$). 
If $s \models r_{i,\ell}$, we put $\Phi_s \equiv \opg(r_{i,\ell} {\vee} a_{i,\ell} \vee \bar{a}_{i,\ell})$ and $\Psi_s \equiv \opg(r_{i,\ell} {\vee} b_{i,\ell})$. Similarly, if $s \models R_i$, we put $\Phi_s \equiv \opg(R_i {\vee} A_i)$ and $\Psi_s \equiv \opg(R_i {\vee} B_i)$. 
Furthermore, $\cv[s] \in [0,1]^2$ denotes the vector such that $\cv[s]_1$ and $\cv[s]_2$ are the probabilities of satisfying $\Phi_s$ and $\Psi_s$ in $s$, respectively.
We define $T$ as the set of all $\cv[s]$ where $s$ is a relevant state. We need to ensure that $T$ satisfies the conditions of Theorem~\ref{thm-area}.
Let

\begin{eqnarray*}
    \Zero & \equiv & 
    \bigwedge_{i=0}^2  \bigwedge_{\ell=1}^m   
    \bigg(r_{i,\ell} \Rightarrow \big(\opg_{=\vec{z}_1}(r_{i,\ell} {\vee} a_{i,\ell} {\vee} \bar{a}_{i,\ell}) \wedge \opg_{=\vec{z}_2}(r_{i,\ell} {\vee} b_{i,\ell}) \big) \bigg)\\
     & \wedge &   \bigwedge_{i=0}^2 
    \bigg(R_i \Rightarrow \big(\opg_{=\vec{z}_1}(R_i {\vee} A_i) \wedge \opg_{=\vec{z}_2}(R_i {\vee} B_i) \big) \bigg)\\[2ex]
    \Eligible & \equiv & \bigwedge_{i=0}^2  \bigwedge_{\ell=1}^m 
    \opg_{=1} \bigg( 
    r_{i,\ell} \Rightarrow  \big(\opg_{\leq \vec{z}_1}(r_{i,\ell} {\vee} a_{i,\ell} {\vee} \bar{a}_{i,\ell}) 
    \wedge \opg_{\geq \beta}(r_{i,\ell} {\vee} a_{i,\ell} {\vee} \bar{a}_{i,\ell})\big) \bigg)\\
    & \wedge & \bigwedge_{i=0}^2 
    \opg_{=1} \bigg(
    R_i \Rightarrow  \big(\opg_{\leq \vec{z}_1}(R_i {\vee} A_i) 
    \wedge \opg_{\geq \beta}(R_i {\vee} A_i)\big)
    \bigg)
\end{eqnarray*}
where $\beta$ is the lower extreme point of $I_\lambda$ (recall that $\lambda$ is chosen so that $\beta$ is rational). Hence, for every relevant state $s$, we have that $s \models \Zero$ iff $\cv[s] = \vec{z}$. The formula $\Eligible$ says that $\cv[s]_1 \in ((1{-}\sqrt{1{-}4\lambda}/2),\vec{z}_1)$ for every relevant~$s$, i.e., $T \subseteq ((1{-}\sqrt{1{-}4\lambda})/2,\vec{z}_1) \times (0,1)$. The closure property of Theorem~\ref{thm-area} is enforced by the formulae constructed in the next paragraph.

\subsection{The subformula $\Init$}
\label{app-init}

We put 
\begin{eqnarray*}
    \Init & \equiv & \ex{r_{0,1}}_{\A} \ \wedge\  \Zero \ \wedge\ \Eligible \ \wedge \ \opg_{=1} \Decrement
\end{eqnarray*}
The subformula $\ex{r_{0,1}}_{\A} \wedge \Zero$ says that a state satisfying $\varphi_\M$ represents the initial configuration $(1,0)$ of~$\M$. The subformulae $\Eligible$ and $\opg_{=1} \Decrement$, together with $\Step_{\ell,\ell'}$ constructed below, ensure that the set $T$ introduced in the previous section satisfies the assumptions of Theorem~\ref{thm-area} (as we shall see in Section~\ref{app-sim-subformula}, $\Decrement$ is similar to the formula implementing the $\Dec$ operation on a counter).

We put
\begin{eqnarray*}
    \Decrement & \equiv & \bigwedge_{i=0}^2 (R_i {\wedge} \neg\Zero) 
        \Rightarrow \bigg( \Copy_i \wedge \Succ_i \bigg)               
\end{eqnarray*}
where
\begin{eqnarray*}
    \Copy_i & \equiv &
       \opg_{=\delta}(R_i {\vee} A_i {\vee} C_i) 
        \ \wedge \ \opf_{=\delta}(R_{S(i)} {\vee} C_i)\\
    \Succ_i & \equiv & \opf_{=\lambda}(B_{S(i)} {\vee} C_{S(i)} {\vee} D_{S(i)} {\vee}
          R_{S^2(i)}) \ \wedge \ \opf_{=\lambda}(B_i {\vee} C_{S(i)} {\vee} D_{S(i)} {\vee} R_{S^2(i)})             
\end{eqnarray*}%
To explain the meaning of $\Decrement$, suppose that $s \models R_i \wedge \Decrement$, where $s$ is a state reachable from a state satisfying $\varphi_\M$ (hence, $s \models \Struct$). Let $U$ be the set of all states $u$ such that $u \models R_{S(i)} \wedge K$ and $p_u > 0$, where $p_u$ is the probability of all runs $w$ initiated in $s$ such that $w$ visits $u$ and all states preceding this visit satisfy $R_i$ (see Fig.~\ref{fig-dec-inc-app}). The formula $\Decrement$ ensures that if $s \models \neg\Zero$, then $\DEC(\cv[s]) = \sum_{u \in U} (p_u/P) \cdot \cv[u]$, where $P = \sum_{u \in U} p_u$. Hence, the condition of Theorem~\ref{thm-area} is satisfied for $\cv[s]$.

More specifically, the subformula $\Copy_i$ ``copies'' $\cv[s]_1$ into $P$, i.e., ensures that $P = \cv[s]_1$. We claim that 
\begin{itemize}
    \item $P$ is the probability of satisfying the path formula $\opf R_{S(i)}$ in~$s$;
    \item $\cv[s]_1$ is the probability of satisfying the path formula $\opg(R_i {\vee} A_i)$ in~$s$.
\end{itemize}
Both claims follow directly from the formula $\Struct$. However, the equality $P = \cv[s]_1$ is not directly expressible in PCTL. Instead, the formula $\Copy_i$ requires that $\cv[s]_1 + \kappa = \delta$ and $P + \kappa = \delta$, where $\kappa$~is the probability of satisfying the formula $\opg(R_i {\vee} C_i)$ in~$s$. 
Note that
\begin{itemize}
    \item $P + \kappa$ is the probability of satisfying the formula $\opf(R_{S(i)} {\vee} C_i)$ in $s$.
    \item $\cv[s]_1 + \kappa$ is the probability of satisfying the formula $\opg(R_i {\vee} A_i {\vee} C_i)$ in $s$. 
\end{itemize}

The first conjunct of $\Succ_i$ says that the probability of satisfying $\opf(B_{S(i)} {\vee} C_{S(i)} {\vee} D_{S(i)} {\vee} R_{S^2(i)})$ in $s$ is equal to~$\lambda$. By inspecting the structure of states reachable from~$s$, we obtain that the above formula is satisfied with the probability $\sum_{u \in U} p_u \cdot (1{-}\cv[u]_1)$. Hence, $\sum_{u \in U} p_u \cdot (1{-}\cv[u]_1) = \lambda$ which implies $(P{-}\lambda)/P = \sum_{u \in U}(p_u/P) \cdot \cv[u]_1$. Since $P = \cv[s]_1$, we have that 
\[
  \DEC(\cv[s])_1 \ = \ (\cv[s]_1 {-} \lambda)/\cv[s]_1 = \sum_{u \in U} (p_u/P) \cdot \cv[u]_1\,.
\]   

The second conjunct of $\Succ_i$ says that the probability of satisfying 
$\opf(B_i {\vee} C_{S(i)} {\vee} D_{S(i)} {\vee} R_{S^2(i)})$ in $s$ is also equal to $\lambda$, i.e., it is the \emph{same} as the probability of satisfying $\opf(B_{S(i)} {\vee} C_{S(i)} {\vee} D_{S(i)} {\vee} R_{S^2(i)})$. From this we obtain that
the following two probabilities are also equal:
\begin{itemize}
    \item The probability of all runs $\pi$ initiated in $s$ visiting a state satisfying $B_i$ so that all states in $\pi$ preceding this visit do not satisfy $R_{S^2(i)}$. 
    \item The probability of all runs $\pi$ initiated in $s$ visiting a state satisfying 
    $B_{S(i)}$  so that all states in $\pi$ preceding this visit do not satisfy $R_{S^2(i)}$. 
\end{itemize}
The first probability is equal to $\cv[s]_2$, and the second probability is equal to $\sum_{u \in U}p_u \cdot \cv[u]_2$. Hence, $\cv[s]_2 =  \sum_{u \in U}p_u \cdot \cv[u]_2$. Since $P = \cv[s]_1$ (see above), we obtain
\[
   \DEC(\cv[s])_2 \ = \ \frac{\cv[s]_2}{\cv[s]_1} \ = \ \frac{\sum_{u \in U}p_u \cdot \cv[u]_2}{P}
   \ = \ \sum_{u \in U} (p_u/P) \cdot \cv[u]_2\,.
\]
To sum up, $\DEC(\cv[s]) =  \sum_{u \in U} (p_u/P) \cdot \cv[u]$ is a convex combination of $\cv[u]$ where $u \in U$. 

\begin{figure}[t]\centering
    \begin{tikzpicture}[x=1.5cm, y=1cm,font=\scriptsize]
        \foreach \x/\y in {0/A,1/B,2/C,3/D,4/E}{
            \node (a\x) at (\x,0) [state] {$\pmb{\y_i}$}; 
        }
        \node[inner sep=0pt,outer sep=0pt,minimum size=1mm] (t1) at (5,0) {$\bullet$};     
        \node[label={[label distance=-2mm]below left:{$u$}},inner sep=0pt,outer sep=0pt,minimum size=1mm] (ti) at (5.5,0) {$\bullet$};   
        \node[inner sep=0pt,outer sep=0pt,minimum size=1mm] (tn) at (6.2,0) {};  
        \node at (6.8,0) {$R_{S(i)},\pmb{K}$};  
        \node[state] at (0,-1.5)   (A1) {$\pmb{A_{S(i)},K}$};
        \node[state] at (1.5,-1.5) (B1) {$\pmb{B_{S(i)},K}$};    
        \node[state] at (3,-1.5)   (C1) {$\pmb{C_{S(i)},K}$};    
        \node[state] at (4.5,-1.5) (D1) {$\pmb{D_{S(i)},K}$};    
        \node[state] at (6,-1.5)   (R1) {$R_{S^2(i)},\pmb{K}$};      
        \node[above right = .3 of tn] {$U$};
        \draw[dotted,thick] ($(t1) +(.15,0)$) -- ($(ti) -(.1,0)$);
        \draw[dotted,thick] ($(ti) +(.15,0)$) -- ($(tn) -(.1,0)$);
        \draw[thick,rounded corners] (4.6,-.27) rectangle (7.4,.25);
        \node[label={[label distance=-1mm]above:{$s \models R_i$}},inner sep=0pt,outer sep=0pt,minimum size=1mm] (s) at (3.5,1.5) {$\bullet$};
        \draw[-stealth,rounded corners,thick] (s) -| node[left=-1mm, near end] {$\cv[s]_1$} (a0);   
        \draw[-stealth,rounded corners,thick] (s) -| node[left=-1mm, near end] {$\cv[s]_2$} (a1);
        \draw[-stealth,rounded corners,thick] (s) -|  (a2);    
        \draw[-stealth,rounded corners,thick] (s) -|  (a3);
        \draw[-stealth,rounded corners,thick] (s) -| node[left, near end] {$\lambda$} (a4);
        \node[inner sep=0pt,outer sep=0pt,minimum size=0pt] (p) at (5.5,1) {};
        \draw[-stealth,rounded corners,thick] (s) -| node[right, near end] {$P = \cv[s]_1$} (p);
        \draw[-stealth,rounded corners,thick] (p) -| (t1);    
        \draw[-stealth,rounded corners,thick] (p) -| node[left=-1mm, near end] {$p_u$} (ti);  
        \draw[-stealth,rounded corners,thick] (p) -| (tn);
    
        \draw[-stealth,rounded corners,thick] (ti) -- ($(ti) -(0,.5)$) -| node[left=0mm, near end]  {$\cv[u]_1$} (A1);       
        \draw[-stealth,rounded corners,thick] (ti) -- ($(ti) -(0,.5)$) -| node[left=0mm, near end]  {$\cv[u]_2$} (B1);    
        \draw[-stealth,rounded corners,thick] (ti) -- ($(ti) -(0,.5)$) -|  (C1);    
        \draw[-stealth,rounded corners,thick] (ti) -- ($(ti) -(0,.5)$) -|  (D1);    
        \draw[-stealth,rounded corners,thick] (ti) -- ($(ti) -(0,.5)$) -|  (R1);    
    \end{tikzpicture}
    \caption{Illustrating the meaning of $\Decrement$.}
    \label{fig-dec-inc-app}
    \end{figure}
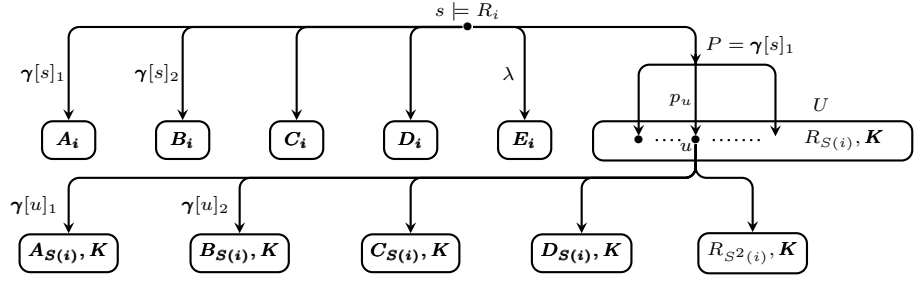

\subsection{The subformula $\Simulate_\ell$}
\label{app-sim-subformula}
Recall that
\begin{equation*}\textstyle
    \Simulate_\ell ~~\equiv~~ (\Zero \Rightarrow \bigvee_{\ell' \in Z^\ell} \Step_{\ell,\ell'}) ~~\wedge~~
    (\neg\Zero \Rightarrow \bigvee_{\ell' \in P^\ell} \Step_{\ell,\ell'}) 
\end{equation*}
where
\begin{eqnarray*}
    \Step_{\ell,\ell'} & \equiv & \bigwedge_{i=0}^2 \ex{r_{i,\ell}}_{\A} \Rightarrow  \left( \ex{r_{i,\ell}}_{\A} \opu_{=1} \rSuc_{i,\ell,\ell'} 
    \ \wedge\ \Update_{i,\ell,\ell'} \right)
\end{eqnarray*}
The purpose of the $\ex{r_{i,\ell}}_{\A} \opu_{=1} \rSuc_{i,\ell,\ell'}$ subformula of $\Step_{\ell,\ell'}$ is to ensure that the $\ell'$ index is ``propagated'' to the states reachable from a state satisfying $\Step_{\ell,\ell'}$ (recall that $\Struct$ ensures the propagation of \emph{some} counter index which may be different from $\ell'$).

To explain the meaning of $\Step_{\ell,\ell'}$ and the formula $\Update_{i,\ell}$ constructed below, let us assume that $s \models \Step_{\ell,\ell'}$ where $s$ is a state reachable form a state satisfying $\varphi_\M$ (in particular, $s \models \Struct$). We distinguish two subcases.


\paragraph{$\pmb{\Ins_\ell}$ updates the counter by $\Dec$}
In this case,  $\Update_{i,\ell,\ell'}$ is similar to $\Decrement$. 
We define the set $U$ of all states $u$ such that $u \models r_{S(i),\ell'}$ and $p_u > 0$, where $p_u$ is the probability of all runs $\pi$ initiated in $s$ such that $\pi$ visits $u$ and all states preceding this visit satisfy $r_{i,\ell}$. 

The formula $\Update_{i,\ell,\ell'}$ ensures that if $s \models \neg\Zero$, then $\DEC(\cv[s])$ is a positive convex combination of $\cv[u]$ where $u \in U$. In addition, it
ensures that if $s \models \Zero$, then $\vec{z}$ is a positive convex combination of $\cv[u]$ where $u \in U$. We put $\Update_{i,\ell,\ell'} \equiv \UDec_{i,\ell,\ell'}$ where

\begin{eqnarray}
    \UDec_{i,\ell,\ell'} & \equiv &  \Zero 
    \Rightarrow \left(\opg_{=\vec{z}_1}(r_{i,\ell} \vee r_{S(i),\ell'} \vee a_{S(i),\ell'})  \wedge  
        \opg_{=\vec{z}_2}(r_{i,\ell} \vee r_{S(i),\ell'} \vee b_{S(i),\ell'})\right)\\
     & \wedge &  \neg\Zero 
        \Rightarrow \left( \UCopy_{i,\ell,\ell'} \wedge \USucc_{i,\ell,\ell'} \right)               
\end{eqnarray}
where
\begin{eqnarray*}
\UCopy_{i,\ell,\ell'} & \equiv &
   \opg_{=\delta}(r_{i,\ell} \vee a_{i,\ell} \vee \bar{a}_{i,\ell} \vee c_{i,\ell}) 
    \ \wedge \ \opf_{=\delta}(r_{S(i),\ell'} \vee c_{i,\ell})\\[1ex]
\USucc_{i,\ell,\ell'} & \equiv & \textstyle\opf_{=\lambda}(b_{S(i),\ell'} \vee c_{S(i),\ell'} \vee d_{S(i),\ell'} \vee \bigvee_{\ell'' = 1}^m r_{S^2(i),\ell''})\\
 & \wedge & \textstyle\opf_{=\lambda}(b_{i,\ell} \vee c_{S(i),\ell'} \vee d_{S(i),\ell'} \vee \bigvee_{\ell'' = 1}^m r_{S^2(i),\ell''})             
\end{eqnarray*}%
By using similar arguments as in Section~\ref{app-init}, it is easy to verify that 
\begin{itemize}
    \item if $s \models \UDec_{i,\ell,\ell'} \wedge \neg\Zero$, then $\DEC(\cv[s]) =  \sum_{u \in U} (p_u/P) \cdot \cv[u]$;
    \item if $s \models \UDec_{i,\ell,\ell'} \wedge \Zero$, then $\vec{z} =  \sum_{u \in U} (p_u/P) \cdot \cv[u]$.
\end{itemize}

\paragraph{$\pmb{\Ins_\ell}$ updates the counter by $\Inc$}

Let $V$ be the set consisting of all $v$ such that $v \models a_{i,\ell} \vee \bar{a}_{i,\ell}$ and $q_v > 0$, where $q_v$ is the probability of all runs $\pi$ initiated in $s$ such that $\pi$ visits $v$ and all states preceding this visit satisfy~$r_{i,\ell}$. Furthermore, let $W$ be the set consisting of all $w$ such that \mbox{$w \models b_{i,\ell} \vee c_{i,\ell} \vee d_{i,\ell} \vee r_{S(i),\ell'}$} and $q_w > 0$, where $q_w$ is the probability of all runs $\pi$ initiated in $s$ such that $\pi$ visits $w$ and all states preceding this visit satisfy $r_{i,\ell}$.
The formula $\Update_{i,\ell,\ell'}$ ensures the following:
\begin{itemize}
    \item $\INC(\cv[s])$ is a positive convex combination of $\cv[w]$ where $w \in W$.
    \item If $s \models \neg\Zero$, then $\DEC(\cv[s])$ is a positive convex combination of $\cv[v]$ where $v \in V$.
\end{itemize}
We put $\Update_{i,\ell,\ell'} \equiv \UInc_{i,\ell,\ell'}$ where $\UInc_{i,\ell,\ell'}$ is the following formula:
{\small%
\begin{eqnarray}
     &  & \textstyle 
       \opg_{=\lambda} \big( (r_{i,\ell} \vee \bigvee_{j=0}^2 R_j \vee r_{S(i),\ell'} \vee
              a_{S(i),\ell'} \vee \bar{a}_{S(i),\ell'} \vee \bigvee_{j=0}^2 A_j)  \wedge \neg K \wedge \neg a_{i,\ell} \wedge \neg\bar{a}_{i,\ell} \big)\label{incsub-one}\\
    & \wedge & \textstyle 
       \opg_{=\varrho}  \big( (r_{i,\ell} \vee \bigvee_{j=0}^2 R_j \vee r_{S(i),\ell'} \vee
       \bar{a}_{i,\ell} \vee \bigvee_{j=0}^2 B_j \vee b_{S(i),\ell'})  \wedge \neg a_{i,\ell} \wedge (K \Rightarrow \bar{a}_{i,\ell})\big)\label{incsub-two}\\
    & \wedge & \textstyle 
       \opg_{=\varrho} \big( (r_{i,\ell} \vee \bigvee_{j=0}^2 R_j \vee  \bar{a}_{i,\ell} \vee
       \bigvee_{j=0}^2 (E_j {\wedge} b_{i,\ell}))   \wedge (K \Rightarrow \bar{a}_{i,\ell}) \big)\label{incsub-three}\\
    & \wedge & \textstyle
       \neg\Zero \Rightarrow 
           \opg_{=\lambda} \big( r_{i,\ell} \vee a_{i,\ell}
              \vee \bar{a}_{i,\ell}) \wedge \bigwedge_{j=0}^2 (A_j \Rightarrow K) \big)\label{incsub-four}\\
    & \wedge &\textstyle
      \neg\Zero \Rightarrow 
       \opg_{=\delta}\big( (r_{i,\ell} \vee c_{i,\ell} \vee \bigvee_{j=0}^2 (R_j {\wedge} a_{i,\ell})
       \vee \bigvee_{j=0}^2 (R_j {\wedge} \bar{a}_{i,\ell}) \vee \bigvee_{j=0}^2 B_j) \wedge (K \Rightarrow c_{i,\ell)} \big)\label{incsub-five}\\
    & \wedge &\textstyle
       \neg\Zero \Rightarrow   
       \opg_{=\delta}(r_{i,\ell} \vee b_{i,\ell} \vee c_{i,\ell})\label{incsub-six}     
\end{eqnarray}}%
The meaning of $\UInc_{i,\ell,\ell'}$ is illustrated in Fig.~\ref{fig-update-inc}.  Suppose $s \models \UInc_{i,\ell,\ell'}$. Every subformula of $\UInc_{i,\ell,\ell'}$ specifies the probability of visiting a state of a certain family by a run initiated in $s$.  In Fig.~\ref{fig-update-inc}, these families are indicated by colored shapes (some states belong to multiple families). 
\begin{itemize}
       \item The subformula~\eqref{incsub-one} says that the probability of reaching a state of the \emph{yellow circle} family is equal to $\lambda$, i.e., 
          $\sum_{w \in W} p_w \cdot \cv[w]_1 = \lambda$. Since $P_W = \sum_{w \in W} p_w = (1-\cv[s]_1)$, we obtain that
          \[
             \INC(\cv[s])_1 \ = \ 
             \frac{\lambda}{1-\cv[s]_1} \ = \ 
             \frac{\sum_{w \in W} p_w \cdot \cv[w]_1}{P_W} \ = \
             \sum_{w \in W} (p_w/P_W) \cdot \cv[w]_1
          \]
       \item The subformulae~\eqref{incsub-two}~and~\eqref{incsub-three} say that the probabilities of reaching a state of the \emph{blue square} and the \emph{red star} families are equal to $\varrho$. Consequently, $\sum_{w \in W} p_w \cdot \cv[w]_2 = \cv[s]_2 \cdot \lambda$ (this also explains the purpose of the subformula $\FLambda$, see Section~\ref{sec-structure}).
       This implies
          \[
             \INC(\cv[s])_2 \ = \ 
             \frac{\cv[s]_2 \cdot \lambda}{1-\cv[s]_1} \ = \ 
             \frac{\sum_{w \in W} p_w \cdot \cv[w]_2}{P_W} \ = \
             \sum_{w \in W} (p_w/P_W) \cdot \cv[w]_2
          \]
       \item If $s \models \neg\Zero$, then the subformula~\eqref{incsub-four} says that the probability of reaching a state of the \emph{green diamond} family is equal to $\lambda$. Hence, $\sum_{v \in V} p_v \cdot (1-\cv[v]_1) = \lambda$. Since $P_V = \sum_{v \in V} p_v = \cv[s]_1$, we obtain
          \[
             \DEC(\cv[s])_1 \ = \ 
             \frac{\cv[s]_1 - \lambda}{\cv[s]_1} \ = \ 
             \frac{\cv[s]_1 - \sum_{v \in V} p_v \cdot (1-\cv[v]_1)}{P_V} \ = \ 
             \sum_{v \in V} (p_v/P_V) \cdot \cv[v]_1
          \]
       \item If $s \models \neg\Zero$, then the subformulae~\eqref{incsub-five}
       and~\eqref{incsub-six} say that the probabilities of reaching a state of the \emph{brown triangle} and the \emph{black semicircle} family are equal to $\delta$. This implies
       $\sum_{v \in V} p_v \cdot \cv[v]_2 = \cv[s]_2$. Hence,
       \[
            \DEC(\cv[s])_2 \ = \ 
            \frac{\cv[s]_2}{\cv[s]_1} \ = \ 
            \frac{\sum_{v \in V} p_v \cdot \cv[v]_2}{P_V} \ = \ 
              \sum_{v \in V} (p_v/P_V) \cdot \cv[v]_2 \,.
       \]
\end{itemize}

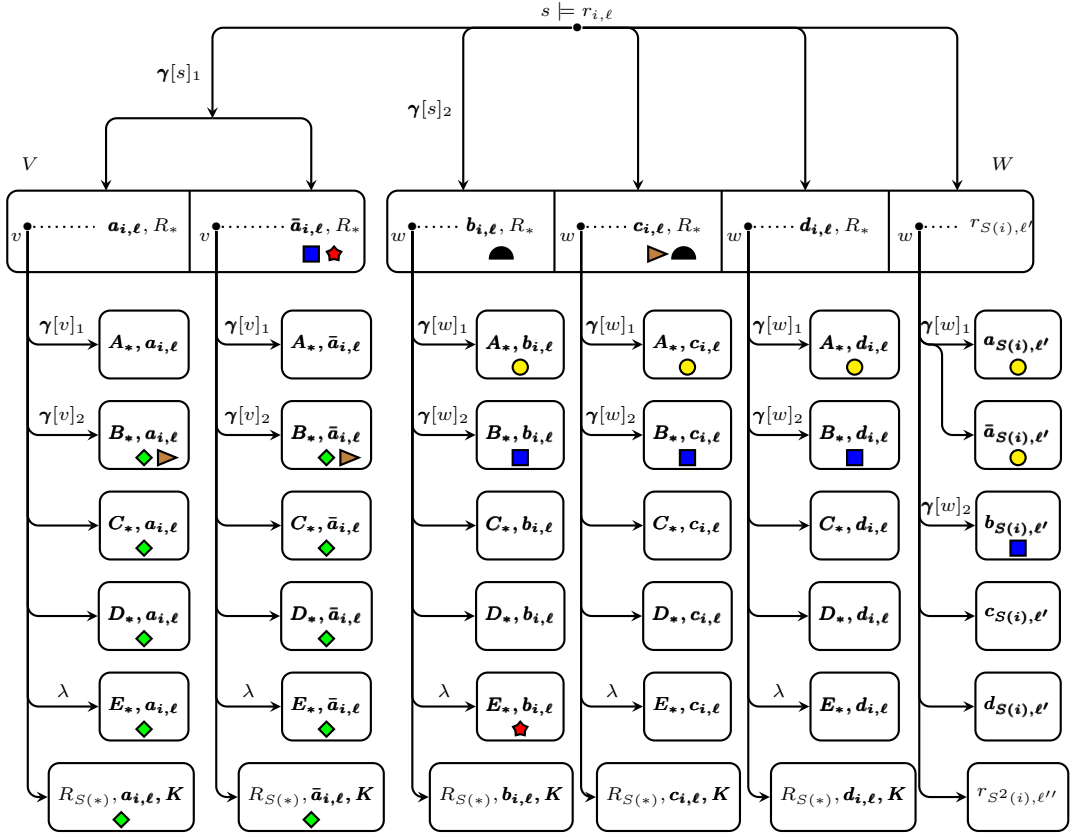
\begin{figure}[t]\centering
    \begin{tikzpicture}[x=1cm, y=1.2cm, font=\scriptsize]
       \tikzstyle{state}=[draw,thick,minimum width=7mm,minimum height=9mm,rounded corners,text centered]
       \tikzstyle{marker}=[circle,draw,thick,fill=yellow,inner sep=0pt,outer sep=0pt,minimum size=2mm]
        \node[label={[label distance=-1mm]above:{$s \models r_{i,\ell}$}},inner sep=0pt,outer sep=0pt,minimum size=1mm] (s) at (6.5,2.5) {$\bullet$}; 
        \draw[thick,rounded corners] (-1,-.2) rectangle (3.7,.7);
        \node[label={[label distance=-1.5mm]below left:{$v$}},inner sep=0pt,outer sep=0pt,minimum size=1mm] (v1) at (-.73,.3) {$\bullet$};
        \node[draw=none,inner sep=0pt,outer sep=0pt,minimum size=1mm] (v2) at (.2,.3) {};
        \node at (0.8,.3) {$\pmb{a_{i,\ell}},R_*$};
        \node at (-.7,1) {$V$};
        \draw[thick,dotted] (v1) -- (v2);
        \node[label={[label distance=-1.5mm]below left:{$v$}},inner sep=0pt,outer sep=0pt,minimum size=1mm] (V1) at (1.75,.3) {$\bullet$};
        \node[draw=none,inner sep=0pt,outer sep=0pt,minimum size=1mm] (V2) at (2.7,.3) {};
        \draw[thick,dotted] (V1) -- (V2);
        \node at (3.2,.3) {$\pmb{\bar{a}_{i,\ell}},R_*$};
        \draw[thick] (1.4,-.2) -- (1.4,.7);
        \draw[thick,rounded corners] (4,-.2) rectangle (12.5,0.7);
        \node[label={[label distance=-1.5mm]below left:{$w$}},inner sep=0pt,outer sep=0pt,minimum size=1mm] (b1) at (4.33,.3) {$\bullet$}; 
        \node[label={[label distance=-1.5mm]below left:{$w$}},inner sep=0pt,outer sep=0pt,minimum size=1mm] (c1) at (6.55,.3) {$\bullet$};
        \node[label={[label distance=-1.5mm]below left:{$w$}},inner sep=0pt,outer sep=0pt,minimum size=1mm] (d1) at (8.75,.3) {$\bullet$};
        \node[label={[label distance=-1.5mm]below left:{$w$}},inner sep=0pt,outer sep=0pt,minimum size=1mm] (r1) at (11,.3) {$\bullet$};
        \node[draw=none,inner sep=0pt,outer sep=0pt,minimum size=1mm] (b2) at (5,.3) {};
        \node[draw=none,inner sep=0pt,outer sep=0pt,minimum size=1mm] (c2) at (7.2,.3) {};
        \node[draw=none,inner sep=0pt,outer sep=0pt,minimum size=1mm] (d2) at (9.4,.3) {};
        \node[draw=none,inner sep=0pt,outer sep=0pt,minimum size=1mm] (r2) at (11.6,.3) {};
        \node at (5.5,.3) {$\pmb{b_{i,\ell}},R_*$};
        \node at (7.7,.3) {$\pmb{c_{i,\ell}},R_*$};
        \node at (9.9,.3) {$\pmb{d_{i,\ell}},R_*$};
        \node at (12.1,.3) {$r_{S(i),\ell'}$};
        \draw[thick,dotted] (b1) -- (b2);
        \draw[thick,dotted] (c1) -- (c2);
        \draw[thick,dotted] (d1) -- (d2);
        \draw[thick,dotted] (r1) -- (r2);
        \draw[thick] (6.2,-.2) -- (6.2,.7);
        \draw[thick] (8.4,-.2) -- (8.4,.7);
        \draw[thick] (10.6,-.2) -- (10.6,.7);
        \node at (12.1,1) {$W$};
        \node[inner sep=0pt,outer sep=0pt,minimum size=0pt] (aa) at (1.7,1.5) {};
        \draw[tran] (s) -| node[left,near end] {$\cv[s]_1$} (aa);
        \draw[tran] (aa) -| (.3,0.7);
        \draw[tran] (aa) -| (3,0.7);
        \draw[tran] (s)  -| node[left,near end] {$\cv[s]_2$} (5,.7);
        \draw[tran] (s)  -| (7.3,.7);
        \draw[tran] (s)  -| (9.5,.7);
        \draw[tran] (s)  -| (11.5,.7);
        \foreach \m/\n/\x in {a/a/0.8,t/\bar{a}/3.2,b/b/5.75,c/c/7.95,d/d/10.15}{
            \foreach \l/\y in {A/-1,B/-2,C/-3,D/-4,E/-5}{
                \node[state] at (\x,\y) (\m\l) {$\pmb{\l_{*},\n_{i,\ell}}$};     
            }
        }
        \foreach \m/\n/\x in {a/a/0.5,t/\bar{a}/3,b/b/5.5,c/c/7.7,d/d/10}{
              \node[state] at (\x,-6) (\m R) {$R_{S(*)},\pmb{\n_{i,\ell},K}$};  
        }
        \node[state] at (12.3,-1) (ra)  {$\pmb{a_{S(i),\ell'}}$};
        \node[state] at (12.3,-2) (raa) {$\pmb{\bar{a}_{S(i),\ell'}}$};
        \node[state] at (12.3,-3) (rb)  {$\pmb{b_{S(i),\ell'}}$};
        \node[state] at (12.3,-4) (rc)  {$\pmb{c_{S(i),\ell'}}$};
        \node[state] at (12.3,-5) (rd)  {$\pmb{d_{S(i),\ell'}}$};
        \node[state] at (12.3,-6) (rr)  {$r_{S^2(i),\ell''}$};
        \foreach \x in {b,c,d}{
            \foreach \y/\n in {A/{$\cv[w]_1$},B/{$\cv[w]_2$},C/{},D/{},R/{}}{
                \draw[tran] (\x1) |- node[above, near end] {\n}(\x\y);
            }
        }
        \draw[tran] (v1)  |- node[above,near end] {$\cv[v]_1$} (aA);
        \draw[tran] (v1)  |- node[above,near end] {$\cv[v]_2$} (aB);
        \draw[tran] (v1)  |- (aC);
        \draw[tran] (v1)  |- (aD);
        \draw[tran] (v1)  |- node[above,near end] {$\lambda$} (aE);
        \draw[tran] (v1)  |- (aR);
        \draw[tran] (V1)  |- node[above,near end] {$\cv[v]_1$} (tA);
        \draw[tran] (V1)  |- node[above,near end] {$\cv[v]_2$} (tB);
        \draw[tran] (V1)  |- (tC);
        \draw[tran] (V1)  |- (tD);
        \draw[tran] (V1)  |- node[above,near end] {$\lambda$} (tE);
        \draw[tran] (V1)  |- (tR);
        \draw[tran] (r1)  |- node[above,near end] {$\cv[w]_1$} (ra);
        \draw[tran] (r1)  |- ($(ra) +(-1,0)$) |- (raa);,c
        \draw[tran] (r1)  |- node[above,near end] {$\cv[w]_2$} (rb);
        \draw[tran] (r1)  |- (rc);
        \draw[tran] (r1)  |- (rd);
        \draw[tran] (r1)  |- (rr);
        \foreach \x in {b,c,d}{
              \draw[tran] (\x1) |- node[above,near end] {$\lambda$} (\x E);  
        }
        \tikzstyle{marker}=[draw,thick,inner sep=0pt,outer sep=0pt,minimum size=2mm]
        \foreach \n in {bA,cA,dA,ra,raa}{
           \node[marker,circle,fill=yellow] at ($(\n) +(0,-.25)$) {};
        }
        \foreach \n in {bB,cB,dB,rb}{
              \node[marker,fill=blue] at ($(\n) +(0,-.25)$) {};
        }
        \node[marker,fill=blue] at (3,0) {};
        \foreach \n in {bE}{
              \node[marker,star,fill=red] at ($(\n) +(0,-.25)$) {};
        }
        \node[marker,star,fill=red] at (3.3,0) {};
        \foreach \n in {aB,aC,aD,aE,aR,tB,tC,tD,tE,tR}{
              \node[marker,diamond,fill=green] at ($(\n) +(0,-.25)$) {};
        }        
        \foreach \n in {aB,tB}{
              \node[marker,isosceles triangle, align=left, fill=brown] at ($(\n) +(0.25,-.25)$) {};
        }
        \node[marker,isosceles triangle, align=center, fill=brown] at (7.5,0) {};
        \node[marker,semicircle,minimum size=1.5mm,fill=black] at (7.9,0) {};
        \node[marker,semicircle,minimum size=1.5mm,fill=black] at (5.5,0) {};
    \end{tikzpicture}
    \caption{Illustrating the meaning of $\Update_{i,\ell,\ell'}$ when $\Ins_\ell$ updates the counter by $\Inc$. The markers are in boldface, and `*' indicates an index $j \in \{0,1,2\}$. The colored shapes indicate the families of states used in the subformulae of $\UInc_{i,\ell,\ell'}$. Note that the probability of certain transitions is equal to $\lambda$ due to the subformula $\FLambda$.}
\label{fig-update-inc}
\end{figure}

\subsection{Justifying the correctness of $\varphi_\M$}
We claim that $\varphi_\M$ is satisfiable iff $\M$ has a recurrent computation, i.e., 
\begin{itemize}
    \item if $t \models \varphi_\M$, then there is a run $\pi$ initiated in $t$ representing a recurrent computation of~$\M$ (here we use the results of Section~\ref{sec-geom});
    \item for every recurrent computation of $\M$, there is a model of $\varphi_\M$ containing a run representing the computation. 
\end{itemize}
We do not give explicit proofs for these claims since they are simplified versions of the proofs of Propositions~\ref{prop-correct} and~\ref{prop-model} formulated in Section~\ref{sec-twocounter}. 


\section{Extension to Non-Deterministic Two-Counter Machines}
\label{sec-twocounter}

Now we show how to adapt the construction of $\varphi_\M$ presented in Section~\ref{sec-formula} when $\M$ is a non-deterministic \emph{two-counter} machine with $m$~instructions.

Let $\M^1,\M^2$ be non-deterministic one-counter machines obtained from $\M$ by changing every instruction of form
\[
    \langle C_k{=}0\, {?}\ Z \mbox{ : } P \rangle:\ \update_1, \update_2
\]
into 
\[
    \langle C{=}0\, {?}\ Z \mbox{ : } P \rangle:\ \update
\]
where $\update \equiv \update_1$ for $\M^1$ and $\update \equiv \update_2$ for $\M^2$. Observe that $\M^1$ and $\M^2$ do \emph{not} simulate $\M$ in any reasonable sense.

Consider the formulae $\varphi_{\M^1}$ and $\varphi_{\M^2}$ constructed for $\M^1$ and $\M^2$ in the way described in the previous sections, where the underlying sets $\A^1$ and $\A^2$ of atomic propositions are disjoint. For every subformula $\Form$ constructed in the previous sections, we use $\Form^1$ and $\Form^2$ to denote the corresponding formulae of  $\varphi_{\M^1}$ and $\varphi_{\M^2}$, respectively.

Intuitively, the formula $\varphi_\M$ is ``basically'' the conjunction of $\varphi_{\M^1}$ and $\varphi_{\M^2}$ with some additional synchronization. First, for every $\ell \in \{1,\ldots,m\}$, let $\NewZero_\ell$ be the formula defined as follows:
\begin{itemize}
    \item If the counter tested for zero in $\Ins_\ell$ is $C_1$, then $\NewZero_\ell \equiv \Zero^1$.
    \item If the counter tested for zero in $\Ins_\ell$ is $C_2$, then $\NewZero_\ell \equiv \Zero^2$.
\end{itemize}  
Recall that $\varphi_{\M^1}$ and $\varphi_{\M^2}$ contain the subformulae 
$\Simulate_\ell^1$ and $\Simulate_\ell^2$ corresponding to the subformula
\begin{equation*}\textstyle
    \Simulate_\ell ~~\equiv~~ (\Zero \Rightarrow \bigvee_{\ell' \in Z^\ell} \Step_{\ell,\ell'}) ~~\wedge~~
    (\neg\Zero \Rightarrow \bigvee_{\ell' \in P^\ell} \Step_{\ell,\ell'}) 
\end{equation*} 
defined in Section~\ref{app-sim-subformula}. Let $\psi_{\M^1}$ and $\psi_{\M^2}$ be the formulae obtained from $\varphi_{\M^1}$ and $\varphi_{\M^2}$ by replacing the 
two $\Simulate_\ell^1$ and $\Simulate_\ell^2$ with $\NewSimulate_\ell^1$ and $\NewSimulate_\ell^2$, where
{\small%
\begin{eqnarray*}
    \NewSimulate_\ell^1 & \equiv & \textstyle\big((\NewZero_\ell \wedge \at_{\ell}^2) \Rightarrow
       \bigvee_{\ell' \in Z^\ell} \Step^1_{\ell,\ell'}\big) ~\wedge~
    \big((\neg\NewZero_\ell \wedge \at_{\ell}^2) \Rightarrow \bigvee_{\ell' \in P^\ell} \Step^1_{\ell,\ell'}\big)\\
     & \wedge & 
    \bigwedge_{i=0}^2 \left(\ex{r_{i,\ell}^1}_{\A^1} \wedge \neg\at_\ell^2\right) \Rightarrow  \left( \ex{r_{i,\ell}^1}_{\A^1} \opu_{=1} \rSuc^1_{i,\ell,\ell} 
    \ \wedge\ \UDec^1_{i,\ell,\ell} \right)\\  
    \NewSimulate_\ell^2 & \equiv & \textstyle\big((\NewZero_\ell \wedge \at_{\ell}^1) \Rightarrow \bigvee_{\ell' \in Z^\ell} \Step^2_{\ell,\ell'}\big) 
    ~\wedge~
    \big((\neg\NewZero_\ell \wedge \at_{\ell}^1) \Rightarrow \bigvee_{\ell' \in P^\ell} \Step^2_{\ell,\ell'}\big)\\
  & \wedge & \bigwedge_{i=0}^2 \left(\ex{r_{i,\ell}^2}_{\A^2} \wedge \neg\at_\ell^1\right) \Rightarrow  \left( \ex{r_{i,\ell}^2}_{\A^2} \opu_{=1} \rSuc^2_{i,\ell,\ell} \ \wedge\ \UDec^2_{i,\ell,\ell} \right) 
\end{eqnarray*}}%
Intuitively, replacing $\Zero$ with $\NewZero_\ell$ ensures that the counter tested for zero in $\NewSimulate_\ell$ is indeed the counter tested for zero in $\Ins_\ell$.
The reason for adding the third conjunct in $\NewSimulate_\ell^1$ and $\NewSimulate_\ell^2$ is more subtle. Roughly speaking, we cannot prevent the situation when a state satisfies $\at_\ell^1$ but not $\at_\ell^2$ (or vice versa). In this case, it does not make sense to continue the simulation of $\M$. However, we still need to ensure that the assumptions of Theorem~\ref{thm-area} are satisfied. Therefore, we start to decrement the counter $C_1$ (or $C_2$).

Observe that $\psi_{\M^1}$ and $\psi_{\M^2}$ may still ``choose'' a different target label when simulating $\Ins_\ell$. Furthermore, even if they choose the same target label, the simulation of the counter's updates is performed completely independently. Hence, there is no guarantee that a state $t$ encoding a configuration $(\ell,c_1,c_2)$ of $\M$ can reach a state encoding a successor configuration $(\ell',c'_1,c'_2)$. This is enforced by the formula 
\begin{eqnarray*}
    \Sync & \equiv & \bigwedge_{\ell=1}^m \bigwedge_{i=0}^2
        \opg_{=1} \left((r^1_{i,\ell} {\wedge} r^2_{i,\ell}) \Rightarrow\textstyle
         \bigvee_{\ell' =1}^m \opg_{>0} \big( (r^1_{i,\ell} {\wedge} r^2_{i,\ell}) 
         \vee (r^1_{S(i),\ell'} {\wedge} r^2_{S(i),\ell'}) \vee a^1_{S(i),\ell'}
         \big)\right)
\end{eqnarray*}
Technically, $\Sync$ says that whenever a state $t$ satisfying $r^1_{i,\ell} \wedge r^2_{i,\ell}$ is visited, then $t$ has a successor $t'$ satisfying $r^1_{S(i),\ell'} \wedge r^2_{S(i),\ell'}$ for some $\ell'$. Due to the subformulae $\Struct^1$ and $\Struct^2$, the state $t'$ must be visited form $t$ via a path such that all states except for $t'$ satisfy $r^1_{i,\ell} \wedge r^2_{i,\ell}$. Note that $t'$ has a successor satisfying the marker $a^1_{S(i),\ell'}$. Hence, the probability of all runs initiated in $t$ satisfying the formula
\[
  \opg \big( (r^1_{i,\ell} {\wedge} r^2_{i,\ell}) 
\vee (r^1_{S(i),\ell'} {\wedge} r^2_{S(i),\ell'}) \vee a^1_{S(i),\ell'}\big)
\]
is positive (we could use other markers instead of $a^1_{S(i),\ell'}$). 

We put 
\begin{equation*}
    \varphi_{\M} ~\equiv~ \psi_{\M^1} \wedge \psi_{\M^2} \wedge \Sync \wedge \Recurrent
\end{equation*}
where 
\begin{eqnarray*}
    \Recurrent & \equiv & \textstyle\opg_{=1}\big(\bigwedge_{\ell=1}^m (\at^1_\ell {\wedge} \at^2_\ell) \Rightarrow \opf_{>0} (\bigvee_{\ell' \in \tau}(\at^1_{\ell'} {\wedge} \at^2_{\ell'}))\big)
\end{eqnarray*}
The formula $\Recurrent$ enforces the existence of a run representing a $\tau$-recurrent computation of~$\M$.

We say that a run $\pi = s_0,s_1,\ldots$ of a Markov chain represents a computation $\conf_0,\conf_1,\ldots$ of $\M$ if there is an infinite increasing sequence of indexes $j_0,j_1,\ldots$ such that for every $n \in \N$, the state $s_{j_n}$ represents the configuration $\conf_n$, i.e., if $\conf_n = (\ell,c_1,c_2)$, then $s_{j_n} \models (\at^1_\ell \wedge \at^2_\ell)$, $\cv^1[s_{j_n}] = \INC^{c_1}(\vec{z})$, and $\cv^2[s_{j_n}] = \INC^{c_2}(\vec{z})$.

\begin{proposition}
\label{prop-correct}
    If $s \models \varphi_\M$, then there exists $\pi \in \run(s)$ representing a recurrent computation of $\M$.
\end{proposition}
\begin{proof}
    Let $\calR^1 = \{r^1_{i,\ell},R^1_i \mid 0 \leq i \leq 2, 1 \leq \ell \leq m\}$ and $\calR^2 = \{r^2_{i,\ell},R^2_i \mid 0 \leq i \leq 2, 1 \leq \ell \leq m\}$. For a given state $t$ and $k \in \{1,2\}$, we say that $t$ is \emph{$k$-relevant} if $t$ satisfies some proposition of $\calR^k$. For a $k$-relevant state $t$, we define $\cv^k[t]$ as the pair of probabilities of satisfying the path formulae $\Phi_t^k$ and $\Psi_t^k$ in $t$, respectively (see Section~\ref{sec-formula}).
    
    Let us fix a state $s$ such that $s \models \varphi_\M$.
    We put 
    \[
       T = \left\{\cv^k[t] \mid k \in \{1,2\}, t \mbox{ is $k$-relevant and reachable from $s$}\right\}
    \]
    It follows directly form the construction of $\varphi_\M$ that $T$ satisfies the assumptions of Theorem~\ref{thm-area}. 

    Consider the sequence of states $s_0,s_1,\ldots$ definite inductively as follows:
    \begin{itemize}
        \item $s_0 = s$. Observe that $s$ encodes the initial configuration $(1,0,0)$ of $\M$.
        \item Consider the state $s_j$ encoding a configuration $(\ell,c_1,c_2)$
        such that $s_j \models r^1_{i,\ell} \wedge r^2_{i,\ell}$. For all $k \in \{1,2\}$, let $U^k$ be the set of all states $u$ such that $u \models r^k_{S(i),\ell'}$ and $p_u > 0$, where $p_u$ is the probability of all runs initiated in $s_j$ such that $w$ 
        visits $u$ and all states preceding this visit satisfy $r_{i,\ell}^k$. It follows from the construction of $\varphi_\M$ that the vector $\vec{v}^k$ representing the counter value obtained from $c_k$ by performing $\update_\ell^k$ is a positive convex combination of a set of vectors 
        $Y \subseteq T$, where $\cv^k[u] \in Y$ for every $u \in U^k$. By applying the results of Section~\ref{sec-geom}, we obtain that $\cv^k[u] = \vec{v}^k$ for every $u \in U^k$. Hence, we can choose $s_{j+1}$ as an (arbitrary) element of $U^1 \cap U^2$ which must be non-empty due to $\Sync$. Note that $s_{j+1}$ encodes a successor configuration of $(\ell,c_1,c_2)$.
    \end{itemize}
    The above sequence $s_0,s_1,\ldots$ is \emph{not} unique, and the runs associated with different sequences may represent different computations of $\M$. At least one of these computations must be $\tau$-recurrent; otherwise, we obtain a contradiction with $s \models \Recurrent$. 
\end{proof}

\begin{proposition}
\label{prop-model}
    For every recurrent computation of $\M$ there exist a Markov chain $M$, a state $s$ of $M$, and $\pi \in \run(s)$ such that $s \models \varphi_\M$ and $\pi$ represents the computation.
\end{proposition}
A proof of Proposition~\ref{prop-model} is relatively simple but technical. 
It can be found in \fp.

\section{The Undecidability Results}

The high undecidability of PCTL satisfiability is an immediate consequence of Propositions~\ref{prop-correct} and~\ref{prop-model}. Note that the 
formula $\varphi_\M$ constructed in in Section~\ref{sec-twocounter} uses only the path connective~$\opu$ and the connectives $\opf$, $\opg$ definable from~$\opu$. Hence, the high undecidability of PCTL satisfiability holds even for the $\opu$-fragment of PCTL. Observe that $\opu$ is actually used only in the subformulae $\NewSimulate_\ell^1$ and $\NewSimulate_\ell^2$. In \fp, it is shown that these formulae can be rewritten so that they use only the connectives $\opf$ and $\opg$. Hence, the result holds even for the $\opf,\opg$-fragment of PCTL.
Furthermore, when we omit the subformula enforcing the existence of a run representing a recurrent computation and assume that $\M$ is deterministic, then the resulting formula has a \emph{finite} model iff $\M$ has a bounded computation. This means that \emph{finite} PCTL satisfiability is undecidable even for the $\opf,\opg$-fragment of PCTL. Thus, we obtain our main result.

\begin{theorem}   
The satisfiability problem for the $\opf,\opg$-fragment of PCTL is $\Sigma_1^1$-hard, and the finite satisfiability problem for the $\opf,\opg$-fragment of PCTL is $\Sigma_1^0$-hard. 
\end{theorem}
Consequently, the validity problem for the $\opf,\opg$-fragment is $\Pi_1^1$-hard, and the finite validity problem for the $\opf,\opg$-fragment is $\Pi_1^0$-hard. This implies that there is no complete deductive system proving all valid (or finitely valid) formulae of the $\opf,\opg$-fragment.

\section{Conclusions}

We have shown that the PCTL satisfiability problem is highly undecidable even for the \mbox{$\opf,\opg$-fragment}. An interesting direction for future research is to characterize the decidability border for PCTL satisfiability, i.e., to establish principle boundaries of automatic probabilistic program synthesis from PCTL specifications.

\bibliography{str-long,concur}

\clearpage
\appendix

\begin{center}
  \huge\bf Appendix
\end{center}
\section{Non-Deterministic $d$-Counter Machines}
\label{app-Minsky}

In this section, we prove Proposition~\ref{prop-twocounter}. We start by recalling the non-deterministic Minsky machines \cite{Minsky:book} and the corresponding undecidability results.

A \emph{non-deterministic Minsky machine $\M$ with $d \geq 1$ counters} is a finite program
\[
    1: \Ins_1; \ \cdots  \ m: \Ins_m;
\]
where $m \geq 1$ and every $i: \Ins_i$ is a labeled instruction of one of the following types:
\begin{itemize}
  \item[I.] $i: \textit{inc } c_j; \textit{ goto } L;$ 
  \item[II.] $i: \textit{if } c_j{=}0 \textit{ then goto } L \textit{ else dec } c_j; \textit{ goto } L'$
\end{itemize}
Here, $j \in \{1,\ldots,d\}$ is a counter index and $L,L' \subseteq \{1,\ldots,m\}$
are sets of labels with one or two elements. We say that $\M$ is \emph{deterministic} if all $L,L'$ occurring in the instructions of $\M$ are singletons\footnote{Our definition of non-deterministic Minsky machines is equivalent to the standard one where the target sets of labels are singletons, and there is also a Type~III instruction of the form $i: \textit{goto } u \textit{ or } u'$. For purposes of this paper, the adopted definition is more convenient.}.

A \emph{configuration} of $\M$ is a tuple $(i,n_1,\ldots,n_k)$ of non-negative integers where $1 \leq i \leq m$ represents the current control position and $n_1,\ldots,n_k$ represent the current counter values. A configuration $(i',n_1',\ldots,n_k')$ is a \emph{successor} of a configuration  $(i,n_1,\ldots,n_k)$, written $(i,n_1,\ldots,n_k) \mapsto (i',n_1',\ldots,n_k')$, if the tuple  $(n_1',\ldots,n_k')$ is obtained from $(n_1,\ldots,n_k)$ by performing $\Ins_i$, and $i'$ is an element of the corresponding $L$ (or $L'$) in $\Ins_i$. Note that every configuration has either one or two successor(s). A \emph{computation} of $\M$ is an infinite sequence of configurations $\omega \equiv C_0,C_1,\ldots$ such that $C_0 = (1,0,\ldots,0)$ and $C_i \mapsto C_{i+1}$ for all $i \in \N$. 

Now, we recall the standard undecidability results for Minsky machines. The symbols $\Sigma_1^0$ and $\Sigma_1^1$ denote the corresponding levels in the arithmetical and the analytical hierarchies, respectively.
\smallskip


(1) The \emph{boundedness problem} for a given deterministic two-counter Minsky machine $\M$ is undecidable and $\Sigma_1^0$-complete \cite{KSCh:Minsky-boundedness-PCS}. Here, $\M$ is bounded if the unique computation $\omega$ contains only finitely many pairwise different configurations.
\smallskip

(2) The \emph{recurrent reachability problem} for a given non-deterministic \mbox{two-counter} Minsky machine $\M$ is highly undecidable and $\Sigma_1^1$-complete \cite{Harel:Infinite-trees-JACM}. Here, the question is whether there exists a \emph{recurrent} computation $\omega$ of~$\M$ such that the instruction $\Ins_1$ is executed infinitely often along~$\omega$.

Observe that every instruction of a Minsky machine updates the value of just one counter and leaves the other counters unchanged. Hence, simulating Minsky machines by PCTL formulae would require inventing some mechanism for ``transferring'' the pairs of probabilities representing positive counter values when moving from one state to another. This is not needed when simulating our non-deterministic $d$-counter machines, because here every instruction updates all counters simultaneously.

\Minsky*

\begin{proof}
    Let $\M \equiv 1:\Ins_1;\cdots m: \Ins_m;$ be a non-deterministic two-counter Minsky machine. The counter modified by $\Ins_\ell$ is \emph{$\ell$-active}, and the other counter is \emph{$\ell$-inactive}.
    
    We construct a non-deterministic two-counter machine $\mathcal{N}$ faithfully simulating~$\M$. The machine $\mathcal{N}$ has $3m$ instructions constructed as follows:
    \begin{itemize}
        \item For every Type~I instruction $i: \textit{inc } c_j; \textit{ goto } L;$ of $\M$, the machine $\mathcal{N}$ has the following instructions:
        \begin{itemize}
            \item An instruction indexed by $i$ of the form 
               \( 
                   \langle C_j{=}0 \,?\, X : X \rangle : \Inc,\Inc
               \)
               where the set $X$ is obtained from $L$ by substituting every $\ell \in L$ with either $\ell + m$ or $\ell +2m$, depending on whether the $\ell$-active counter is $c_j$ or not, respectively.
            \item An instruction indexed by $i {+} m$ of the form
               \( 
                   \langle C_j{=}0 \,?\, L : L \rangle : \update_1,\update_2
               \),
               where $\update_j = \Inc$ and the other update is $\Dec$.
        \end{itemize}
        \item For every Type~II instruction $i: \textit{if } c_j{=}0 \textit{ then goto } L \textit{ else dec } c_j; \textit{ goto } L'$ of $\M$, the machine $\mathcal{N}$ has the following instructions:
        \begin{itemize}
            \item An instruction indexed by $i$ of the form 
               \( 
                   \langle C_j{=}0 \,?\, Z : P \rangle : \update_1,\update_2
               \)
               where the sets $Z$ and $P$ are obtained from $L$ and $L'$ by substituting every $\ell \in L$ (or every $\ell \in L'$) with either $\ell + m$ or $\ell +2m$, depending on whether the $\ell$-active counter is $c_j$ or not, respectively. Furthermore, $\update_j = \Dec$ and the other update is $\Inc$.
            \item An instruction indexed by $i {+} m$ of the form
               \( 
                   \langle C_j{=}0 \,?\, L : L' \rangle : \Dec,\Dec
               \).
        \end{itemize}
        \item Furthermore, for every $i \in \{1,\ldots,m\}$, the machine $\mathcal{N}$ has an instruction indexed by $i {+} 2m$ of the form
        \( 
            \langle C_1{=}0 \,?\, \{i{+}m\} : \{i{+}m\} \rangle : \update_1,\update_2
        \), where $\update_j = \Dec$ and the other update is $\Inc$.
    \end{itemize}
    Intuitively, $\mathcal{N}$ simulates $\M$ by performing the corresponding instruction on the active counter. The other counter is incremented/decremented alternately until the moment when it becomes active. Note that this can also happen in a situation when its value is inconsistent due to the previous increment. Then, an auxiliary decrement is inserted before continuing the simulation. The auxiliary decrements are performed by instructions indexed by $i{+}2m$, where $i \in \{1,\ldots,m\}$. The alternating increments/decrements of the inactive counter are performed by instructions indexed by $i$ and $i+m$ where $i \in \{1,\ldots,m\}$.
    
    Clearly, if $\M$ deterministic and bounded, then $\mathcal{N}$ is deterministic and has a bounded computation. Furthermore, $\M$ has a recurrent computation iff $\mathcal{N}$ has a $\{1,m{+}1\}$-recurrent computation.
\end{proof}

\section{Representing Non-Negative Integers by Points in $(0,1)^2$}
\label{app-area}

In this section, we prove Lemma~\ref{lem-tausigma} and Lemma~\ref{lem-outlineseg}.

\incfunction*
\begin{proof}
\textit{Item~(a).} We show that $\INC(\vec{v}) \in I_\lambda \times [0,1]$. Observe
    \begin{eqnarray*}
        \INC(\vec{v})_1 & = & \frac{\lambda}{1-\vec{v}_1}
           \ < \ \frac{\lambda}{1-\frac{1+\sqrt{1-4\lambda}}{2}}
           \ = \ \frac{2\lambda}{1- \sqrt{1-4\lambda}} \cdot \frac{1+\sqrt{1-4\lambda}}{1+\sqrt{1-4\lambda}}
           \ = \ \frac{1+\sqrt{1-4\lambda}}{2}
    \end{eqnarray*}
    Similarly, we obtain $\INC(\vec{v})_1 >  \frac{1-\sqrt{1-4\lambda}}{2}$. Since $\INC(\vec{v})_2 = \frac{\lambda \cdot \vec{v}_2}{1{-}\vec{v}_1} = \vec{v}_2 \cdot \INC(\vec{v})_1$ and $\INC(\vec{v})_1 \in I_\lambda$, we have that $\INC(\vec{v})_2 \in (0,1)$.

\smallskip  
\noindent
\textit{Items~(b) and (c)} are trivial to verify.
 
\smallskip
\noindent
\textit{Item~(d).}  Observe
\begin{eqnarray*}
    \slope(\vec{u},\INC(\vec{v})) & = & \frac{\INC(\vec{v})_2}{\INC(\vec{v})_1 - \INC^2(\vec{v})_1} = \frac{\frac{\lambda \vec{v}_2}{1-\vec{v}_1}}{\frac{\lambda}{1-\vec{v}_1}-\frac{\lambda(1-\vec{v}_1)}{1-\vec{v}_1 -\lambda}} \ = \
    \frac{\vec{v}_2(1-\vec{v}_1-\lambda)}{\vec{v}_1 - \lambda - \vec{v}_1^2}
 \end{eqnarray*}
Similarly,
\begin{eqnarray*}
    \slope(\INC(\vec{v}),\vec{v}) & = &  
    \frac{\vec{v}_2 - \frac{\lambda \vec{v}_2}{1-\vec{v}_1}}{\vec{v}_1 - \frac{\lambda}{1-\vec{v}_1}} \ = \   
    \frac{\vec{v}_2(1-\vec{v}_1-\lambda)}{\vec{v}_1 - \lambda - \vec{v}_1^2}  
\end{eqnarray*}
Hence, $\slope(\vec{u},\INC(\vec{v})) = \slope(\INC(\vec{v}),\vec{v})$.

\smallskip
\noindent
\textit{Item~(e).} Realize that 
\begin{eqnarray*}    
    \slope(\vec{u},\DEC(\vec{u})) & = & \frac{y(1-\vec{v}_1)}{\lambda} \cdot \frac{1-\vec{v}_1-\lambda}{\vec{v}_1 -\vec{v}_1^2-\lambda}\\[1ex]
    \slope(\INC(\vec{v}),\vec{v}) & = & \vec{v}_2 \cdot \frac{1-\vec{v}_1-\lambda}{\vec{v}_1 -\vec{v}_1^2-\lambda}
\end{eqnarray*}
Since 
\begin{equation*}
    \frac{y(1-\vec{v}_1)}{\lambda} \ < \ \frac{\INC(\vec{v})_2 \cdot (1-\vec{v}_1)}{\lambda}  \ = \ \vec{v}_2   
\end{equation*}
we have that $\slope(\vec{u},\DEC(\vec{u})) < \slope(\INC(\vec{v}),\vec{v})$.

\smallskip
\noindent
\textit{Item~(f).}
It is easy to verify that for all $\vec{x},\vec{y} \in I_\lambda \times (0,1)$ and all $\kappa \in (0,1]$ we have that
\[
    \DEC(\kappa\vec{x} + (1{-}\kappa)\vec{y}) \ = \ 
    \kappa'\DEC(\vec{x}) + (1{-}\kappa') \DEC(\vec{y})
\]
where 
\[
    \kappa' = \frac{\kappa \vec{x}_1}{\kappa\vec{x}_1 + (1{-}\kappa) \vec{y}_1} 
\]
Observe that $\kappa' \in (0,1]$. Item~(f) follows by putting $\vec{x} = \INC^2(\vec{v})$, $\vec{y} = \INC(\vec{v})$ and choosing $\kappa$ so that $\vec{u} = \kappa\vec{x} + (1{-}\kappa)\vec{y}$. 

\smallskip
\noindent
\textit{Item~(g).}

First, we show that 
\begin{eqnarray}
    \lim_{k \to \infty} \INC^k(\vec{v})_1 & = & \frac{1-\sqrt{4q-3}}{2} \label{lim-2}
\end{eqnarray}
Let 
\[
    J_q =    \left[\frac{1-\sqrt{1-4\lambda}}{2},  \vec{v}_1\right] 
\]
By Items~(a)~and~(c), the infinite sequence  $\vec{v}_1, \INC(\vec{v})_1, \INC^2(\vec{v})_1,\ldots$ is decreasing and bounded from below by $(1-\sqrt{1-4\lambda})/2$. Consequently, the sequence has a limit~$\alpha \in J_q$, and hence it is also a Cauchy sequence, i.e.,
\[
    \lim_{n \to \infty} \INC^{n+1}(\vec{v})_1 - \INC^{n}(\vec{v})_1 = 0 \,.
\] 
Consider the function $f : J_q \to \R$ where 
\[
    f(x) \ = \ \frac{\lambda-x(1-x)}{1-x}
\]
Observe that $f$ is non-negative and continuous. Furthermore, $f(x) = 0$ iff $x =  (1 - \sqrt{1-4\lambda})/2$. Observe that for every $n \in \N$, we have that
\begin{eqnarray*}
    \INC^{n+1}(\vec{v})_1 - \INC^{n}(\vec{v})_1 \ = \ 
        f(\INC^{n}(v)_1) \,.
\end{eqnarray*}
Hence, 
\[
    0 \ = \ \lim_{n \to \infty} \INC^{n+1}(\vec{v})_1 - \INC^{n}(\vec{v})_1  =  \lim_{n \to \infty} f(\INC^{n}(\vec{v})_1) \ = \ f(\alpha)
\]
which implies $\alpha = (1 - \sqrt{1-4\lambda})/2$.
\end{proof}

\outlineseg*
\begin{proof}
We choose $\vec{u}_2$ so that 
\[
    \slope(\vec{u},\DEC(\vec{u})) = \slope(\vec{u},\vec{v})\,.
\]
Hence, we require that
\[
     \frac{\vec{u}_2(1-\vec{u}_1)}{\vec{u}_1(1-\vec{u}_1) - \lambda}    =
     \frac{\vec{v}_2 - \vec{u}_2}{\vec{v}_1 - \vec{u}_1}
\]
From this, we obtain
\[
    \vec{u}_2 = \frac{\vec{v}_2(\vec{u}_1(1-\vec{u}_1)-\lambda)}{\vec{v}_1 - \vec{u}_1 +  \vec{u}_1(1-\vec{u}_1) - \lambda}    
\]
and the proof is finished.
\end{proof}

\section{Avoiding the $\opu$-operator}
\label{sec-U-avoid}

The formula $\varphi_\M$ constructed in Section~\ref{sec-formula} uses the path operator~$\opu$. More concretely, all occurrences of $\opu$ in $\varphi_\M$ are within the formula $\Struct$. In this section, we show that using $\opu$ can be avoided, i.e., we design a new formula $\overline{\Struct}$ such that $\overline{\Struct}$ implies $\Struct$ and $\overline{\Struct}$ contains only the $\opf$ and $\opg$ operators.

First, define $\mathcal{C} = \{a_{i,\ell},\bar{a}_{i,\ell},b_{i,\ell},c_{i,\ell},d_{i,\ell} \mid 0\leq i \leq 2, 1\leq \ell \leq m\}$.
Now, define $\overline{\Struct}$ as follows:
\[ \overline{\Struct} \equiv \overline{\Succ} \wedge Mark \wedge Lambda, \]
where $Mark$ and $Lambda$ remain as they were in the original $\Struct$
and $\overline{\Succ}$ is defined as follows:
\begin{align*}
  \overline{\Succ} \equiv &
  \bigwedge_{i=0}^{2}
  \bigwedge_{\ell=1}^{m}
  \opg_{=1}
  \left(
    \langle r_{i,\ell} \rangle_{\A}
    \implies
    \bigvee_{\ell'=1}^{m} \opf_{=1} rsuc_{i,\ell,\ell'}
  \right) \\
  & \wedge \\
  & \bigwedge_{i=0}^{2}
    \bigwedge_{\ell=1}^{m}
    \opg_{=1}
    \left(
      \langle r_{i,\ell} \rangle_{\A}
      \implies
      \opf_{<1}
      \bigvee_{\substack{
                 0 \leq z \leq 2 \\
                 1 \leq z' \leq m \\
                 z \neq i \text{ or } z' \neq \ell}
                }
      \ex{r_{z,z'}}_{\A} 
    \right) \\
  & \wedge \\
  & \opg_{=1} \bigvee_{x \in \A \setminus \B} \ex{x}_{\A \setminus \B} \\
  & \wedge \\
  & \opg_{=1} \bigwedge_{x \in \mathcal{C}} \left( x \implies \bigvee_{0 \leq j \leq 2} \ex{x,R_j}_{\A} \right) \\
  & \wedge \\
  & \opg_{=1} \bigwedge_{x \in \A \setminus (\B \cup \mathcal{C})} \left( x \implies \ex{x}_{\A} \right) \\
  & \wedge \\
  & \bigwedge_{i=0}^{2}
    \opg_{=1}
    \bigg( \ex{R_i}_{\B} \implies \opf_{=1} Rsuc_i \bigg) \\
  & \wedge \\
  & \bigwedge_{i=0}^{2}
    \opg_{=1}
    \bigg( \ex{R_i,K}_{\B} \implies \opf_{=1} RKsuc_i \bigg) \\
  & \wedge \\
  & \bigwedge_{i=0}^2
    \opg_{=1}
    \left(
      \bigg[ \ex{R_i}_{\B} \vee \ex{R_i,K}_{\B} \bigg]
    \implies
    \opf_{<1} \bigvee_{\substack{0 \leq z \leq 2 \\ z \neq i}} \ex{R_{z},K}_\B
    \right) \\
  & \wedge \\
  & \opg_{=1}
    \left(
      \bigg[
        \bigvee_{x \in \B} x
      \bigg]
      \implies
      \bigg[
        \opg_{=1} \big[ \bigvee_{i=0}^2 Rsuc_i \vee RKsuc_i \big]
      \bigg]
    \right) \\
  & \wedge \\
  & \opg_{=1} (K \implies \opg_{=1} K)
\end{align*}

Clearly, $\overline{\Struct}$ is a PCTL state formula 
in the $\opf,\opg$ fragment of PCTL. It remains to show
$\overline{\Struct}$ implies $\Struct$.

Suppose that $\overline{\Struct}$ holds in some state $o$.
We want to show that $\Struct$ also holds in $o$.
Clearly $Mark$ and $Lambda$ are satisfied,
so it remains to show that $\Succ$ is satisfied.
Recall the definition of $\Succ$:
\begin{eqnarray*}\\[-2em]
    \Succ & \equiv & \bigwedge_{\ell=1}^m \bigwedge_{i=0}^2 \opg_{=1} \bigg(  \ex{r_{i,\ell}}_{\A} \ \Rightarrow \ \bigvee_{\ell'=1}^m \ex{r_{i,\ell}}_{\A} \opu_{=1} \rSuc_{i,\ell,\ell'} \bigg)\\ 
       & \wedge & \bigwedge_{i=0}^2 \opg_{=1} \bigg(
       \big( \ex{R_i}_{\B} \ \Rightarrow \ \ex{R_i}_{\B} \opu_{=1} \RSuc_i \big) 
       \wedge
       \big( \ex{R_i,K}_{\B} \ \Rightarrow \ \ex{R_i,K}_{\B} \opu_{=1} \RKSuc_i \big) 
       \bigg)
\end{eqnarray*}
where\\[-2.5em]
\begin{eqnarray*}
    \rSuc_{i,\ell,\ell'} & \equiv & \ex{r_{S(i),\ell'}}_{\A} \vee \bigvee_{j=0}^2 \left(\ex{a_{i,\ell},R_j}_{\A} \ \vee \ex{\bar{a}_{i,\ell},R_j}_{\A}\ \vee \ \ex{b_{i,\ell}, R_j}_{\A} \ \vee \ \ex{c_{i,\ell},R_j}_{\A} \ \vee \ \ex{d_{i,\ell},R_j}_{\A}\right)  \\
    \RSuc_i & \equiv & \ex{A_i}_{\B} \ \vee \ \ex{B_i}_{\B} \ \vee \ \ex{C_i}_{\B} \ \vee \ \ex{D_i}_{\B} \ \vee \ \ex{E_i}_{\B} \ \vee \ \ex{R_{S(i)},K}_{\B}\\
    \RKSuc_i & \equiv & \ex{A_i,K}_{\B} \ \vee \ \ex{B_i,K}_{\B} \ \vee \ \ex{C_i,K}_{\B} \ \vee \ \ex{D_i,K}_{\B}  \ \vee \ \ex{R_{S(i)},K}_{\B}.
\end{eqnarray*} 

We split the proof into two propositions to make it more readable.
In the first proposition, we show that
\[
\bigwedge_{\ell=1}^m \bigwedge_{i=0}^2 \opg_{=1} \bigg(  \ex{r_{i,\ell}}_{\A} \ \Rightarrow \ \bigvee_{\ell'=1}^m \ex{r_{i,\ell}}_{\A} \opu_{=1} \rSuc_{i,\ell,\ell'} \bigg)
\]
holds.
In the second proposition, we show that
\[
       \bigwedge_{i=0}^2 \opg_{=1} \bigg(
       \big( \ex{R_i}_{\B} \ \Rightarrow \ \ex{R_i}_{\B} \opu_{=1} \RSuc_i \big) 
       \wedge
       \big( \ex{R_i,K}_{\B} \ \Rightarrow \ \ex{R_i,K}_{\B} \opu_{=1} \RKSuc_i \big) 
       \bigg)
\]
holds.

\begin{proposition}
The state $o$ satisfies the formula
\[
\bigwedge_{\ell=1}^m \bigwedge_{i=0}^2 \opg_{=1} \bigg(  \ex{r_{i,\ell}}_{\A} \ \Rightarrow \ \bigvee_{\ell'=1}^m \ex{r_{i,\ell}}_{\A} \opu_{=1} \rSuc_{i,\ell,\ell'} \bigg).
\]
\end{proposition}

\begin{proof}
Let $\ell \in \{1, \ldots, m\}$ and $i \in \{ 0, \ldots, 2 \}$.
Suppose $\ex{r_{i,\ell}}_{\A}$ is satisfied in some state $s$ that
is reachable from $o$.
We want to show that $s$ satisfies
\[
\bigvee_{\ell'=1}^m \ex{r_{i,\ell}}_{\A} \opu_{=1} \rSuc_{i,\ell,\ell'}.
\]
We know that $o$ satisfies $\overline{\Succ}$, and since $s$ is a state
reachable from $o$ that satisfies $\ex{r_{i,\ell}}_{\A}$,
it is easy to see that $s$ must satisfy
\[
\bigvee_{\ell'=1}^{m} \opf_{=1} rsuc_{i,\ell,\ell'}.
\]
Then $s$ satisfies $\opf_{=1} rsuc_{i,\ell,\ell'}$ for some
$\ell' \in \{ 1, \ldots, m \}$. Fix such $\ell'$.
We show that $s$ satisfies
\[
\ex{r_{i,\ell}}_{\A} \opu_{=1} \rSuc_{i,\ell,\ell'}.
\]
We do so by showing that every path $\pi$ starting in $s$
that satisfies $\opf rsuc_{i,\ell,\ell'}$ also satisfies
$\ex{r_{i,\ell}}_{\A} \opu \rSuc_{i,\ell,\ell'}$.

Let $\pi$ be a path that starts in $s$ and satisfies
$\opf rsuc_{i,\ell,\ell'}$.
Then there exists $\alpha \in \nat$ such that $\pi(\alpha)$ satisfies
$rsuc_{i,\ell,\ell'}$.
Let $\alpha$ be the least such $\alpha$.
It is now sufficient to show that $\pi(\beta)$ satisfies $\ex{r_{i,\ell}}_{\A}$ for every $\beta \in \nat$ such that $\beta < \alpha$.

Suppose it does not.
Then there exists $\beta \in \nat$ such that $\beta < \alpha$ and $\pi(\beta)$ does not satisfy $\ex{r_{i,\ell}}_{\A}$.
Let $\beta$ be such $\beta$.
Since $o$ satisfies \( \overline{\Succ} \),
we know that $\pi(\beta)$ satisfies \(\ex{x}_{\A \setminus \B}\) for some \(x \in \A \setminus \B\).
Let $x$ be such $x$.

Now we show that $x \not \in \mathcal{C}$, by contradiction.
Suppose that $x \in \mathcal{C}$.
Since $o$ satisfies \( \overline{\Succ} \),
it follows that $\pi(\beta)$ satisfies $\ex{x,R_j}_{\A}$ for some $0 \leq j \leq 2$.
Recall that $x \in \mathcal{C}$ and 
$\mathcal{C} = \{a_{i,\ell},\bar{a}_{i,\ell},b_{i,\ell},c_{i,\ell},d_{i,\ell} \mid 0\leq i \leq 2, 1\leq \ell \leq m\}$.
Then there exist $j,k \in \nat$ such that
$0 \leq j \leq 2$ and $1 \leq k \leq m$ and
$x$ is one of $a_{j,k}$, $\bar{a}_{j,k}$,
$b_{j,k}$, $c_{j,k}$, $d_{j,k}$.
Fix such $j$ and $k$.
It is easy to see that $j \neq i$ or $k \neq \ell$,
otherwise $\pi(\beta)$ would satisfy $rsucc_{i,\ell,\ell'}$, a contradiction.
Recall that $\pi(\beta)$ satisfies \( \ex{x}_{\A \setminus \B} \)
and $x \in \mathcal{C}$.
Observe that $\{x\}$ is a $\A \setminus \B$ marker for every $x \in \mathcal{C}$.
Given that $\pi(\alpha)$ is obviously reachable from $\pi(\beta)$,
it follows that $\pi(\alpha)$ satisfies \( \ex{x}_{\A \setminus \B} \).
We already know that $x$ is one of $a_{j,k}$, $\bar{a}_{j,k}$,
$b_{j,k}$, $c_{j,k}$, $d_{j,k}$ and that $j \neq i$ or $k \neq \ell$.
It follows that $\pi(\alpha)$ does not satisfy $rsucc_{i,\ell,\ell'}$,
a contradiction.
In conclusion, $x \not \in \mathcal{C}$.

We have now established that $x \not \in \mathcal{C}$.
By its definition, $x \in \A \setminus \B$.
Therefore, $x \in \A \setminus (\B \cup \mathcal{C})$.
Since $o$ satisfies $\opg_{=1} \bigwedge_{x \in \A \setminus (\B \cup \mathcal{C})} \left( x \implies \ex{x}_{\A} \right)$ and $\pi(\beta)$ satisfies $x$,
it follows that $\pi(\beta)$ satisfies $\ex{x}_{\A}$.
Since $x \in \A \setminus \B$ and $x \not \in \mathcal{C}$,
$x$ must be $r_{j,k}$ for some $j,k \in \nat$ such that
$0 \leq j \leq 2$ and $1 \leq k \leq m$.
Fix such $j,k$.
We now know that $\pi(\beta)$ satisfies $\ex{r_{j,k}}_{\A}$.
Observe that $j \neq i$ or $k \neq \ell$,
otherwise $\pi(\beta)$ would satisfy $\ex{r_{i,\ell}}_{\A}$,
contradicting the definition of $\beta$.
Observe that $o$ satisfies $\overline{\Succ}$ and $\pi(\beta)$
is reachable from $o$ and $\pi(\beta)$ satisfies $\ex{r_{j,k}}_{\A}$.
It follows that $\pi(\beta)$ satisfies $\opf_{=1} rsuc_{j,k,k'}$ for some $1 \leq k' \leq m$.
Fix such $k'$.

Now we show that $\pi(\beta)$ satisfies
$\opf_{=1} rsuc_{i,\ell,\ell'}$.
Recall that $s$ satisfies $\opf_{=1} rsuc_{i,\ell,\ell'}$.
Also recall that $\pi$ is a path starting in $s$, so $\pi(0) = s$.
From the definition of $\alpha$, we know that $\pi(\gamma)$ does not
satisfy $rsuc_{i,\ell,\ell'}$ for every $\gamma \in \nat$ such
that $\gamma \leq \beta$.
Now it is easy to see that if $\pi(\beta)$ did not satisfy
$\opf_{=1} rsuc_{i,\ell,\ell'}$, neither would $s$, a contradiction.
In conclusion, $\pi(\beta)$ satisfies $\opf_{=1} rsuc_{i,\ell,\ell'}$.

We have now established that $\pi(\beta)$ satisfies
$\opf_{=1} rsuc_{i,\ell,\ell'}$ and
$\opf_{=1} rsuc_{j,k,k'}$.
We also have already established that $j \neq i$ or $\ell \neq k$.
Recall that $\pi(\beta)$ satisfies $\ex{r_{j,k}}_{\A}$.
Recall that $o$ satisfies $\overline{\Succ}$,
and so it follows that $\pi(\beta)$ must also satisfy
\[
    \opf_{<1}
    \bigvee_{\substack{
               0 \leq z \leq 2 \\
               1 \leq z' \leq m \\
               z \neq j \text{ or } z' \neq k}
              }
    \ex{r_{z,z'}}_{\A}.
\]
Let $\pi'$ be a path starting in $\pi(\beta)$ that
satisfies $\opf rsuc_{j,k,k'}$.
Let $\gamma \in \nat$ be the least natural number such that $\pi'$
satisfies $rsuc_{j,k,k'} \vee rsuc_{i,\ell,\ell'}$ (it is easy
to see that such $\gamma$ exists).
It is easy to see that if $\pi'(\gamma)$ satisfied one of
$a_{i,\ell}$, $\bar{a}_{i,\ell}$, $b_{i,\ell}$, $c_{i,\ell}$, $d_{i,\ell}$,
then $\pi'$ could not satisfy $\opf rsuc_{j,k,k'}$, a contradiction.
It is not hard to see that if $\pi'(\gamma)$ satisfied one of
$a_{j,k}$, $\bar{a}_{j,k}$, $b_{j,k}$, $c_{j,k}$, $d_{j,k}$,
then $\pi(\beta)$ could not satisfy $\opf_{=1} rsuc_{i,\ell,\ell'}$,
a contradiction.
Therefore, $\pi'(\gamma)$ satisfies $\ex{r_{S(i),\ell'}}_{\A}
\vee \ex{r_{S(j),k'}}_{\A}$.
%
%
It follows that $\pi'(\gamma)$ satisfies 
\[
    \bigvee_{\substack{
               0 \leq z \leq 2 \\
               1 \leq z' \leq m \\
               z \neq j \text{ or } z' \neq k}
              }
    \ex{r_{z,z'}}_{\A}.
\]
Therefore, every path starting in $\pi(\beta)$ that
satisfies $\opf rsuc_{j,k,k'}$ also satisfies
\[
    \opf
    \bigvee_{\substack{
               0 \leq z \leq 2 \\
               1 \leq z' \leq m \\
               z \neq j \text{ or } z' \neq k}
              }
    \ex{r_{z,z'}}_{\A}.
\]
We already know that $\pi(\beta)$ satisfies
$\opf_{=1} rsuc_{j,k,k'}$, and so it follows that
$\pi(\beta)$ also satisfies 
\[
  \opf_{=1}
    \bigvee_{\substack{
               0 \leq z \leq 2 \\
               1 \leq z' \leq m \\
               z \neq j \text{ or } z' \neq k}
              }
    \ex{r_{z,z'}}_{\A},
\]
which is a contradiction.
We have now established that
\[
\bigwedge_{\ell=1}^m \bigwedge_{i=0}^2 \opg_{=1} \bigg(  \ex{r_{i,\ell}}_{\A} \ \Rightarrow \ \bigvee_{\ell'=1}^m \ex{r_{i,\ell}}_{\A} \opu_{=1} \rSuc_{i,\ell,\ell'} \bigg)
\]
holds.
\end{proof}

\begin{proposition}
The state $o$ satisfies the formula
\[
       \bigwedge_{i=0}^2 \opg_{=1} \bigg(
       \big( \ex{R_i}_{\B} \ \Rightarrow \ \ex{R_i}_{\B} \opu_{=1} \RSuc_i \big) 
       \wedge
       \big( \ex{R_i,K}_{\B} \ \Rightarrow \ \ex{R_i,K}_{\B} \opu_{=1} \RKSuc_i \big) 
       \bigg).
\]
\end{proposition}

\begin{proof}
Let $i \in \{0, \ldots, 2 \}$.
Let $s$ be some state reachable from the state $o$.
We want to show that 
\[
\big( \ex{R_i}_{\B} \ \Rightarrow \ \ex{R_i}_{\B} \opu_{=1} \RSuc_i \big) 
\wedge
\big( \ex{R_i,K}_{\B} \ \Rightarrow \ \ex{R_i,K}_{\B} \opu_{=1} \RKSuc_i \big) 
\]
holds in $s$.
We start with the first conjunct
\[
\ex{R_i}_{\B} \ \Rightarrow \ \ex{R_i}_{\B} \opu_{=1} \RSuc_i.
\]
Suppose that $\ex{R_i}_{\B}$ holds in $s$.
We want to show that $\ex{R_i}_{\B} \opu_{=1} \RSuc_i$ holds as well.
Since $o$ satisfies $\overline{\Succ}$ and $s$ is reachable
from $o$, we have that $s$ satisfies $\opf_{=1} Rsuc_i$.
To show that $s$ satisfies $\ex{R_i}_{\B} \opu_{=1} \RSuc_i$,
it is sufficient to show that every path starting in $s$
that satisfies $\opf Rsuc_i$ also satisfies 
$\ex{R_i}_{\B} \opu \RSuc_i$.
Let $\pi$ be a path that starts in $s$ and satisfies $\opf Rsuc_i$.
Then there exists $\alpha \in \nat$ such that $\pi(\alpha)$
satisfies $Rsuc_i$.
Let $\alpha$ be the least such $\alpha$.
We want to show that $\pi$ satisfies $\ex{R_i}_{\B} \opu \RSuc_i$.
Suppose it does not.
Then there exists $\beta \in \nat$ such that $\beta < \alpha$
and $\pi(\beta)$ does not satisfy $\ex{R_i}_{\B}$.
Since $s$ is reachable from $o$ and $o$ satisfies $\overline{\Succ}$,
we know that $s$ satisfies
\[
    \bigg[
      \bigvee_{x \in \B} x
    \bigg]
    \implies
    \bigg[
      \opg_{=1} \big[ \bigvee_{i=0}^2 Rsuc_i \vee RKsuc_i \big]
    \bigg].
\]
We already know that $s$ satisfies $R_i$ and clearly $R_i \in \B$.
It follows that $\pi(\beta)$ satisfies $Rsuc_j \vee RKsuc_j$
for some $j \in \{0, \ldots, 2 \}$.
Fix such $j$.
Now we show that $\pi(\beta)$ must satisfy $\ex{R_{S(j)},K}_\B$.

Suppose that $\pi(\beta)$ satisfies one of $\ex{A_j}_{\B}$,
$\ex{B_j}_{\B}$, $\ex{C_j}_{\B}$, $\ex{D_j}_{\B}$,
$\ex{E_j}_{\B}$.
If $j=i$, then $\pi(\beta)$ satisfies $Rsuc_i$, a contradiction.
Otherwise $j \neq i$.
Recall that the sets $\{A_j\}$, $\{B_j\}$, $\{C_j\}$, $\{D_j\}$,
and $\{E_j\}$ are $\B$ markers.
It follows that $\pi(\alpha)$ also satisfies one of
$\ex{A_j}_{\B}$, $\ex{B_j}_{\B}$, $\ex{C_j}_{\B}$, $\ex{D_j}_{\B}$,
$\ex{E_j}_{\B}$.
Then $\pi(\alpha)$ does not satisfy $Rsuc_i$, a contradiction.

Suppose that $\pi(\beta)$ satisfies one of $\ex{A_j,K}_{\B}$,
$\ex{B_j,K}_{\B}$, $\ex{C_j,K}_{\B}$, $\ex{D_j,K}_{\B}$.
Recall that the sets $\{A_j,K\}$, $\{B_j,K\}$, $\{C_j,K\}$,
$\{D_j,K\}$, are $\B$ markers.
It follows that $\pi(\alpha)$ satisfies one of
$\ex{A_j,K}_{\B}$, $\ex{B_j,K}_{\B}$, $\ex{C_j,K}_{\B}$,
$\ex{D_j,K}_{\B}$.
Clearly, this contradicts the fact that $\pi(\alpha)$ satisfies
$Rsuc_i$.

We already know that $\pi(\beta)$ satisfies $Rsuc_j \vee RKsuc_j$,
which means that the only remaining possibility is that $\pi(\beta)$
satisfies $\ex{R_{S(j)},K}_\B$.

Observe that this also means that $\pi(\beta)$ satisfies $K$.
Since $o$ satisfies $\overline{\Succ}$ and $\pi(\beta)$ is
reachable from $o$, we know that $\pi(\beta)$ satisfies 
$K \implies \opg_{=1} K$.
Therefore, all states reachable from $\pi(\beta)$ satisfy $K$
as well.

Recall that $s$ satisfies $\opf_{=1} Rsuc_i$.
By the definition of $\alpha$, $\pi(\gamma)$ does not satisfy
$Rsuc_i$ for every $\gamma \in \nat$ such that $\gamma < \alpha$.
Recall that $\beta \in \nat$ and $\beta < \alpha$.
It follows that $\pi(\beta)$ must satisfy $\opf_{=1} Rsuc_i$.

Since $\pi(\beta)$ satisfies $\opf_{=1} Rsuc_i$
and $\pi(\beta)$ satisfies $K$
and every state reachable from $\pi(\beta)$ satisfies $K$,
we have that $\pi(\beta)$ satisfies $\opf_{=1} \ex{R_{S(i)},K}_{\B}$.

Now we show that $i \neq j$.
Suppose $i=j$.
We have shown that $\pi(\beta)$ satisfies $\ex{R_{S(j)},K}_\B$,
which means that $\pi(\beta)$ satisfies $\ex{R_{S(i)},K}_\B$.
Then $\pi(\beta)$ satisfies $Rsuc_i$, and that contradicts
the definition of $\alpha$.

We already know that $\pi(\beta)$ satisfies
$\opf_{=1} \ex{R_{S(i)},K}_{\B}$.
Given that $i \neq j$, clearly $S(i) \neq S(j)$.
Then it is easy to see that $\pi(\beta)$ also satisfies
\[
\opf_{=1} \bigvee_{\substack{0 \leq z \leq 2 \\ z \neq S(j)}} \ex{R_{z},K}_\B.
\]

Recall that $o$ satisfies $\overline{\Succ}$.
Also recall that $\pi(\beta)$ satisfies $\ex{R_{S(j)},K}_\B$.
It follows that $\pi(\beta)$ satisfies
\[
\opf_{<1} \bigvee_{\substack{0 \leq z \leq 2 \\ z \neq S(j)}} \ex{R_{z},K}_\B,
\]
a contradiction.

It remains to show that the second conjunct
\[
\ex{R_i,K}_{\B} \ \Rightarrow \ \ex{R_i,K}_{\B} \opu_{=1} \RKSuc_i
\]
holds in $s$.
In this case, we proceed analogously to the previous conjunct and
the previous proposition.
\end{proof}

\section{A Proof of Proposition~\ref{prop-model}}
\label{app-two-counter-extension}

Recall that $\lambda = \frac{14}{255}$, $I_\lambda =(\frac{1}{15},\frac{14}{15})$, $\vec{z} = (\frac{1}{12},\frac{1}{15})$, and $\delta = \frac{1}{11}$.  
We also put $\varrho = \frac{1}{13}$. This implies
\begin{eqnarray}
\vec{z}_2 < \varrho < \vec{z}_1 <\delta \label{in-cons-one}\\
2\lambda + 2\delta + 2\vec{z}_1 + 2\vec{z}_2 < 1 \label{in-cons-two}
\end{eqnarray}

\noindent
Now we restate the proposition.
\newtheorem{innercustomprop}{Proposition}
\newenvironment{customprop}[1]
  {\renewcommand\theinnercustomprop{#1}\innercustomprop}
  {\endinnercustomprop}

\begin{customprop}{\ref{prop-model}}
    For every recurrent computation of $\M$ there exist a Markov chain $M$, a state $s$ of $M$, and $\pi \in \run(s)$ such that $s \models \varphi_\M$ and $\pi$ represents the computation.
\end{customprop}
\begin{proof}
   Let $\conf_0,\conf_1,\ldots$ be a recurrent computation of~$\M$. We construct a Markov chain $M$ where the states are elements of $\N \times 2^{\A^1 \cup \A^2} \times \N \times \N$. The first component models the configuration index, the second components is the set of atomic propositions satisfied in the state, and the last two integers represent the counters' values. 

   We say that $\alpha \subseteq \A^1 \cup \A^2$ is \emph{free} if 
   $\alpha \cap \{R^k_i, r^k_{i,\ell} \mid k \in \{1,2\}, i \in \{0,1,2\}, \ell \in \{1,\ldots,m\}\} = \emptyset$. The Markov chain $M$ is the least fixed-point of the closure rules described below. Here we use the constraints~\eqref{in-cons-one} and~\eqref{in-cons-two}. 
   \begin{itemize}
    \item For every $\conf_j = (\ell,c_1,c_2)$, there is a state $[j,r^1_{i,\ell},r^2_{i,\ell},c_1,c_2]$ where $i = j \mbox{ mod } 3$. 
    \item If $[j,r^1_{i,\ell},r^2_{i,\ell},c_1,c_2]$ is a state, then its immediate successors are the associated probabilities are the following:
    \begin{itemize}
        \item If $\update_\ell^1 = \Dec$ and $\update_\ell^2 = \Dec$, then
           \begin{itemize}
            \item $[j,r^1_{i,\ell},r^2_{i,\ell},c_1,c_2] \tran{\vv{c_1}_1} 
                   [j,a^1_{i,\ell},d^2_{i,\ell},R_i^1,R_i^2,0,0]$
            \item $[j,r^1_{i,\ell},r^2_{i,\ell},c_1,c_2] \tran{\vv{c_1}_2} 
                   [j,b^1_{i,\ell},d^2_{i,\ell},R_i^1,R_i^2,0,0]$    
            \item $[j,r^1_{i,\ell},r^2_{i,\ell},c_1,c_2] \tran{\delta- \vv{c_1}_1}
                   [j,c^1_{i,\ell},d^2_{i,\ell},R_i^1,R_i^2,0,0]$
            \item $[j,r^1_{i,\ell},r^2_{i,\ell},c_1,c_2] \tran{\vv{c_2}_1} 
                   [j,d^1_{i,\ell},a^2_{i,\ell},R_i^1,R_i^2,0,0]$
            \item $[j,r^1_{i,\ell},r^2_{i,\ell},c_1,c_2] \tran{\vv{c_2}_2} 
                   [j,d^1_{i,\ell},b^2_{i,\ell},R_i^1,R_i^2,0,0]$    
            \item $[j,r^1_{i,\ell},r^2_{i,\ell},c_1,c_2] \tran{\delta- \vv{c_2}_1} 
                   [j,d^1_{i,\ell},c^2_{i,\ell},R_i^1,R_i^2,0,0]$
        \end{itemize}
        The next transitions depend on the relationship between $\vv{c_1}_1$ and $\vv{c_2}_1$. 
        
        If $\vv{c_1}_1 > \vv{c_2}_1$, then
           \begin{itemize} 
            \item $[j,r^1_{i,\ell},r^2_{i,\ell},c_1,c_2] \tran{p}  [j{+}1,r^1_{S(i),\ell'},r^2_{S(i),\ell'},c'_1,c'_2]$
            \item $[j,r^1_{i,\ell},r^2_{i,\ell},c_1,c_2] \tran{p'}  [j,r^1_{S(i),\ell'},d^2_{i,\ell},R^2_{S(i)},\max\{0,c_1{-}1\},0]$
           \end{itemize}
        where $p = \vv{c_2}_1$, $p' =  \vv{c_1}_1 - \vv{c_2}_1$, and $\conf_{j+1}= (\ell',c'_1,c'_2)$.

        If $\vv{c_1}_1 < \vv{c_2}_1$, then
        \begin{itemize} 
            \item $[j,r^1_{i,\ell},r^2_{i,\ell},c_1,c_2] \tran{p}  [j{+}1,r^1_{S(i),\ell'},r^2_{S(i),\ell'},c'_1,c'_2]$
            \item $[j,r^1_{i,\ell},r^2_{i,\ell},c_1,c_2] \tran{p'}  [j,d^1_{i,\ell},r^2_{S(i),\ell'},R^1_{S(i)},0,\max\{0,c_2{-}1\}]$
           \end{itemize}
        where $p = \vv{c_1}_1$, $p' =  \vv{c_2}_1 - \vv{c_1}_1$, and $\conf_{j+1}= (\ell',c'_1,c'_2)$.

        If $\vv{c_1}_1 = \vv{c_2}_1$, then
        \begin{itemize} 
            \item $[j,r^1_{i,\ell},r^2_{i,\ell},c_1,c_2] \tran{\vv{c_1}_1}  [j{+}1,r^1_{S(i),\ell'},r^2_{S(i),\ell'},c'_1,c'_2]$
           \end{itemize}
        where $\conf_{j+1}= (\ell',c'_1,c'_2)$. Finally,
        \begin{itemize}
            \item $[j,r^1_{i,\ell},r^2_{i,\ell},c_1,c_2] \tran{q} [j,d^1_{i,\ell},d^2_{i,\ell},R_i^1,R_i^2,0,0]$    
         \end{itemize}  
        where $q = 1 - \vv{c_1}_2 - \vv{c_2}_2 - 2\delta - p - p'$.
    \item If $\update_\ell^1 = \Dec$ and $\update_\ell^2 = \Inc$, then
        \begin{itemize}
         \item $[j,r^1_{i,\ell},r^2_{i,\ell},c_1,c_2] \tran{\vv{c_1}_1} 
                [j,a^1_{i,\ell},d^2_{i,\ell},R_i^1,R_i^2,0,c_2{+}1]$
         \item $[j,r^1_{i,\ell},r^2_{i,\ell},c_1,c_2] \tran{\vv{c_1}_2} 
                [j,b^1_{i,\ell},d^2_{i,\ell},R_i^1,R_i^2,0,c_2{+}1]$    
         \item $[j,r^1_{i,\ell},r^2_{i,\ell},c_1,c_2] \tran{\delta- \vv{c_1}_1}
                [j,c^1_{i,\ell},d^2_{i,\ell},R_i^1,R_i^2,0,c_2{+}1]$
         \item $[j,r^1_{i,\ell},r^2_{i,\ell},c_1,c_2] \tran{\vv{c_2}_1 - r} 
                [j,d^1_{i,\ell},a^2_{i,\ell},R_i^1,R_i^2,0,\max\{0,c_2{-}1\}]$
         \item $[j,r^1_{i,\ell},r^2_{i,\ell},c_1,c_2] \tran{r} 
                [j,d^1_{i,\ell},\bar{a}^2_{i,\ell},R_i^1,R_i^2,0,\max\{0,c_2{-}1\}]$
         \item $[j,r^1_{i,\ell},r^2_{i,\ell},c_1,c_2] \tran{\vv{c_2}_2} 
                [j,d^1_{i,\ell},b^2_{i,\ell},R_i^1,R_i^2,0,c_2{+}1]$    
         \item $[j,r^1_{i,\ell},r^2_{i,\ell},c_1,c_2] \tran{\delta- \vv{c_2}_2} 
                [j,d^1_{i,\ell},c^2_{i,\ell},R_i^1,R_i^2,0,c_2{+}1]$
         \item $[j,r^1_{i,\ell},r^2_{i,\ell},c_1,c_2] \tran{\vv{c_1}_1}  
                [j{+}1,r^1_{S(i),\ell'},r^2_{S(i),\ell'},c'_1,c'_2]$
         \item $[j,r^1_{i,\ell},r^2_{i,\ell},c_1,c_2] \tran{q} [j,d^1_{i,\ell},d^2_{i,\ell},R_i^1,R_i^2,0,c_2{+}1]$    
     \end{itemize}
     where $r = \varrho - (\vv{c_2+1}_2 \cdot (1-\vv{c_2}_1))$,
     $q = 1 - \vv{c_1}_2 - \vv{c_2}_1 - \vv{c_1}_1 - 2\delta$,
     and $\conf_{j+1}= (\ell',c'_1,c'_2)$.
    \item If $\update_\ell^1 = \Inc$ and $\update_\ell^2 = \Dec$, then 
     \begin{itemize}
        \item $[j,r^1_{i,\ell},r^2_{i,\ell},c_1,c_2] \tran{\vv{c_2}_1} 
               [j,d^1_{i,\ell},a^2_{i,\ell},R_i^1,R_i^2,c_1{+}1,0]$
        \item $[j,r^1_{i,\ell},r^2_{i,\ell},c_1,c_2] \tran{\vv{c_2}_2} 
               [j,d^1_{i,\ell},b^2_{i,\ell},R_i^1,R_i^2,c_1{+}1,0]$    
        \item $[j,r^1_{i,\ell},r^2_{i,\ell},c_1,c_2] \tran{\delta- \vv{c_2}_1}
               [j,d^1_{i,\ell},c^2_{i,\ell},R_i^1,R_i^2,c_1{+}1,0]$
        \item $[j,r^1_{i,\ell},r^2_{i,\ell},c_1,c_2] \tran{\vv{c_1}_1 - r} 
               [j,a^1_{i,\ell},d^2_{i,\ell},R_i^1,R_i^2,\max\{0,c_1{-}1\},0]$
        \item $[j,r^1_{i,\ell},r^2_{i,\ell},c_1,c_2] \tran{r} 
               [j,\bar{a}^1_{i,\ell},d^2_{i,\ell},R_i^1,R_i^2,\max\{0,c_1{-}1\},0]$
        \item $[j,r^1_{i,\ell},r^2_{i,\ell},c_1,c_2] \tran{\vv{c_1}_2} 
               [j,b^1_{i,\ell},d^2_{i,\ell},R_i^1,R_i^2,c_1{+}1,0]$    
        \item $[j,r^1_{i,\ell},r^2_{i,\ell},c_1,c_2] \tran{\delta- \vv{c_1}_2} 
               [j,c^1_{i,\ell},d^2_{i,\ell},R_i^1,R_i^2,c_1{+}1,0]$
        \item $[j,r^1_{i,\ell},r^2_{i,\ell},c_1,c_2] \tran{\vv{c_2}_1}  
               [j{+}1,r^1_{S(i),\ell'},r^2_{S(i),\ell'},c'_1,c'_2]$
        \item $[j,r^1_{i,\ell},r^2_{i,\ell},c_1,c_2] \tran{q} [j,d^1_{i,\ell},d^2_{i,\ell},R_i^1,R_i^2,c_1{+}1,0]$    
     \end{itemize}
    where $r = \varrho - (\vv{c_1+1}_2 \cdot (1-\vv{c_1}_1))$,
    $q = 1 - \vv{c_2}_2 - \vv{c_1}_1 - \vv{c_2}_1 - 2\delta$,
    and $\conf_{j+1}= (\ell',c'_1,c'_2)$.
    \item If $\update_\ell^1 = \Inc$ and $\update_\ell^2 = \Inc$, then 
    \begin{itemize}
        \item $[j,r^1_{i,\ell},r^2_{i,\ell},c_1,c_2] \tran{\vv{c_1}_1 - r_1} 
               [j,a^1_{i,\ell},d^2_{i,\ell},R_i^1,R_i^2,\max\{0,c_1{-}1\},c_2{+}1]$
        \item $[j,r^1_{i,\ell},r^2_{i,\ell},c_1,c_2] \tran{r_1} 
               [j,\bar{a}^1_{i,\ell},d^2_{i,\ell},R_i^1,R_i^2,\max\{0,c_1{-}1\},c_2{+}1]$
        \item $[j,r^1_{i,\ell},r^2_{i,\ell},c_1,c_2] \tran{\vv{c_1}_2} 
               [j,b^1_{i,\ell},d^2_{i,\ell},R_i^1,R_i^2,c_1{+}1,c_2{+}1]$    
        \item $[j,r^1_{i,\ell},r^2_{i,\ell},c_1,c_2] \tran{\delta- \vv{c_1}_2} 
               [j,c^1_{i,\ell},d^2_{i,\ell},R_i^1,R_i^2,c_1{+}1,c_2{+}1]$
        \item $[j,r^1_{i,\ell},r^2_{i,\ell},c_1,c_2] \tran{\vv{c_2}_1 - r_2} 
                [j,d^1_{i,\ell},a^2_{i,\ell},R_i^1,R_i^2,c_1{+}1,\max\{0,c_2{-}1\}]$
        \item $[j,r^1_{i,\ell},r^2_{i,\ell},c_1,c_2] \tran{r_2} 
                [j,d^1_{i,\ell},\bar{a}^2_{i,\ell},R_i^1,R_i^2,c_1{+}1,\max\{0,c_2{-}1\}]$
        \item $[j,r^1_{i,\ell},r^2_{i,\ell},c_1,c_2] \tran{\vv{c_2}_2} 
                [j,d^1_{i,\ell},b^2_{i,\ell},R_i^1,R_i^2,c_1{+}1,c_2{+}1]$    
        \item $[j,r^1_{i,\ell},r^2_{i,\ell},c_1,c_2] \tran{\delta- \vv{c_2}_2} 
                [j,d^1_{i,\ell},c^2_{i,\ell},R_i^1,R_i^2,c_1{+}1,c_2{+}1]$
        \item $[j,r^1_{i,\ell},r^2_{i,\ell},c_1,c_2] \tran{q}  
               [j{+}1,r^1_{S(i),\ell'},r^2_{S(i),\ell'},c'_1,c'_2]$
        \item $[j,r^1_{i,\ell},r^2_{i,\ell},c_1,c_2] \tran{q} [j,d^1_{i,\ell},d^2_{i,\ell},R_i^1,R_i^2,c_1{+}1,c_2{+}1]$    
     \end{itemize}
     where $r_1 = \varrho - (\vv{c_1+1}_2 \cdot (1-\vv{c_1}_1))$, $r_2 = \varrho - (\vv{c_2+1}_2 \cdot (1-\vv{c_2}_1))$,
     $2q = 1 - \vv{c_2}_2 - \vv{c_1}_1 - \vv{c_2}_1 - 2\delta$,
     and $\conf_{j+1}= (\ell',c'_1,c'_2)$.
    \end{itemize}
    \item For every state $[j,\alpha,r^1_{i,\ell},R_{i}^2,c_1,0]$ where $\alpha \subseteq \A^1 \cup \A^2$ is free, the immediate successors are the associated probabilities are the following:
    \begin{itemize}
        \item $[j,\alpha,r_{i,\ell}^1,R_i^2,c_1,0] \tran{\vv{c_1}_1}
               [j,\alpha,a_{i,\ell}^1,R_i^1,D_i^2,0,0]$ 
        \item $[j,\alpha,r_{i,\ell}^1,R_i^2,c_1,0] \tran{\vv{c_1}_2}
               [j,\alpha,b_{i,\ell}^1,R_i^1,D_i^2,0,0]$ 
        \item $[j,\alpha,r_{i,\ell}^1,R_i^2,c_1,0] \tran{\delta - \vv{c_1}_1}
               [j,\alpha,c_{i,\ell}^1,R_i^1,D_i^2,0,0]$        
        \item $[j,\alpha,r_{i,\ell}^1,R_i^2,c_1,0] \tran{q}
               [j,\alpha,d_{i,\ell}^1,R_i^1,D_i^2,0,0]$ 
        \item $[j,\alpha,r_{i,\ell}^1,R_i^2,c_1,0] \tran{\vv{c_1}_1}
               [j,\alpha,r_{S(i),\ell}^1,D_i^2,\max\{0,c_1{-}1\},0]$ 
        \item $[j,\alpha,r_{i,\ell}^1,R_i^2,c_1,0] \tran{\vec{z}_1}
               [j,\alpha,d_{i,\ell}^1,R_i^1,A_i^2,0,0]$ 
        \item $[j,\alpha,r_{i,\ell}^1,R_i^2,c_1,0] \tran{\vec{z}_2}
               [j,\alpha,d_{i,\ell}^1,R_i^1,B_i^2,0,0]$ 
        \item $[j,\alpha,r_{i,\ell}^1,R_i^2,c_1,0] \tran{\delta - \vec{z}_1}
               [j,\alpha,d_{i,\ell}^1,R_i^1,C_i^2,0,0]$ 
        \item $[j,\alpha,r_{i,\ell}^1,R_i^2,c_1,0] \tran{\lambda}
               [j,\alpha,d_{i,\ell}^1,R_i^1,E_i^2,0,0]$ 
        \item $[j,\alpha,r_{i,\ell}^1,R_i^2,c_1,0] \tran{\vec{z}_1}
               [j,\alpha,d_{i,\ell}^1,R_{S(i)}^1,R_{S(i)}^2,K^2,0,0]$ 
  \end{itemize}
  where $q = 1 - \vv{c_1}_2 - \vv{c_1}_1 -\vec{z}_2 - \vec{z}_1  - 2\delta -\lambda$.
  \item   For every state $[j,\alpha,R_{i}^1,r^2_{i,\ell},0,c_2]$ where $\alpha \subseteq \A^1 \cup \A^2$ is free, the immediate successors are the associated probabilities are the following:

  \begin{itemize}
      \item $[j,\alpha,r_{i,\ell}^2,R_i^1,0,c_2] \tran{\vv{c_2}_1}
             [j,\alpha,a_{i,\ell}^2,R_i^2,D_i^1,0,0]$ 
      \item $[j,\alpha,r_{i,\ell}^2,R_i^1,0,c_2] \tran{\vv{c_2}_2}
             [j,\alpha,b_{i,\ell}^2,R_i^2,D_i^1,0,0]$ 
      \item $[j,\alpha,r_{i,\ell}^2,R_i^1,0,c_2] \tran{\delta - \vv{c_2}_1}
             [j,\alpha,c_{i,\ell}^2,R_i^2,D_i^1,0,0]$        
      \item $[j,\alpha,r_{i,\ell}^2,R_i^1,0,c_2] \tran{q}
             [j,\alpha,d_{i,\ell}^2,R_i^2,D_i^1,0,0]$ 
      \item $[j,\alpha,r_{i,\ell}^2,R_i^1,0,c_2] \tran{\vv{c_2}_1}
             [j,\alpha,r_{S(i),\ell}^2,D_i^1,0,\max\{0,c_2{-}1\}]$ 
      \item $[j,\alpha,r_{i,\ell}^2,R_i^1,0,c_2] \tran{\vec{z}_1}
             [j,\alpha,d_{i,\ell}^2,R_i^2,A_i^1,0,0]$ 
      \item $[j,\alpha,r_{i,\ell}^2,R_i^1,0,c_2] \tran{\vec{z}_2}
             [j,\alpha,d_{i,\ell}^2,R_i^2,B_i^1,0,0]$ 
      \item $[j,\alpha,r_{i,\ell}^2,R_i^1,0,c_2] \tran{\delta - \vec{z}_1}
             [j,\alpha,d_{i,\ell}^2,R_i^2,C_i^1,0,0]$ 
      \item $[j,\alpha,r_{i,\ell}^2,R_i^1,0,c_2] \tran{\lambda}
             [j,\alpha,d_{i,\ell}^2,R_i^2,E_i^1,0,0]$ 
      \item $[j,\alpha,r_{i,\ell}^2,R_i^1,0,c_2] \tran{\vec{z}_1}
             [j,\alpha,d_{i,\ell}^2,R_{S(i)}^1,R_{S(i)}^2,K^2,0,0]$ 
  \end{itemize}
     where $q = 1 - \vv{c_2}_2 - \vv{c_2}_1 -\vec{z}_2 - \vec{z}_1  - 2\delta -\lambda$.
  \item  For every state  $[j,\alpha,r^1_{i,\ell},c_1,0]$ where $\alpha \subseteq \A^1 \cup \A^2$ is free, the immediate successors are the associated probabilities are the following:
  \begin{itemize}
    \item $[j,\alpha,r_{i,\ell}^1,c_1,0] \tran{\vv{c_1}_1}
           [j,\alpha,a_{i,\ell}^1,R_i^1,0,0]$ 
    \item $[j,\alpha,r_{i,\ell}^1,c_1,0] \tran{\vv{c_1}_2}
           [j,\alpha,b_{i,\ell}^1,R_i^1,0,0]$ 
    \item $[j,\alpha,r_{i,\ell}^1,c_1,0] \tran{\delta - \vv{c_1}_1}
           [j,\alpha,c_{i,\ell}^1,R_i^1,0,0]$        
    \item $[j,\alpha,r_{i,\ell}^1,c_1,0] \tran{q}
           [j,\alpha,d_{i,\ell}^1,R_i^1,0,0]$ 
    \item $[j,\alpha,r_{i,\ell}^1,c_1,0] \tran{\vv{c_1}_1}
           [j,\alpha,r_{S(i),\ell}^1,\max\{0,c_1{-}1\},0]$ 
   \end{itemize}
   where $q = 1 - \vv{c_1}_2 - \vv{c_1}_1 - \delta$.
   \item For every state $[j,\alpha,r^2_{i,\ell},0,c_2]$ where $\alpha \subseteq \A^1 \cup \A^2$ is free, the immediate successors are the associated probabilities are the following:
   \begin{itemize}
     \item $[j,\alpha,r_{i,\ell}^2,0,c_2] \tran{\vv{c_2}_1}
            [j,\alpha,a_{i,\ell}^2,R_i^2,0,0]$ 
     \item $[j,\alpha,r_{i,\ell}^2,0,c_2] \tran{\vv{c_2}_2}
            [j,\alpha,b_{i,\ell}^2,R_i^2,0,0]$ 
     \item $[j,\alpha,r_{i,\ell}^2,0,c_2] \tran{\delta - \vv{c_2}_1}
            [j,\alpha,c_{i,\ell}^2,R_i^2,0,0]$        
     \item $[j,\alpha,r_{i,\ell}^2,0,c_2] \tran{q}
            [j,\alpha,d_{i,\ell}^2,R_i^2,0,0]$ 
     \item $[j,\alpha,r_{i,\ell}^2,0,c_2] \tran{\vv{c_2}_1}
            [j,\alpha,r_{S(i),\ell}^2,0,\max\{0,c_2{-}1\}]$ 
    \end{itemize}
    where $q = 1 - \vv{c_2}_2 - \vv{c_2}_1 - \delta$.

  \item For every state $[j,\alpha,R_i^1,R_i^2,c_1,c_2]$ where $\alpha \subseteq \A^1 \cup \A^2$ is free, the immediate successors are the associated probabilities are the following:
  \begin{itemize}
      \item $[j,\alpha,R_i^1,R_i^2,c_1,c_2] \tran{\vv{c_1}_1}
             [j,\alpha,A_i^1,D_i^2,0,0]$ 
      \item $[j,\alpha,R_i^1,R_i^2,c_1,c_2] \tran{\vv{c_1}_2}
             [j,\alpha,B_i^1,D_i^2,0,0]$ 
      \item $[j,\alpha,R_i^1,R_i^2,c_1,c_2] \tran{\delta - \vv{c_1}_1}
             [j,\alpha,C_i^1,D_i^2,0,0]$        
      \item $[j,\alpha,R_i^1,R_i^2,c_1,c_2] \tran{q}
             [j,\alpha,D_i^1,D_i^2,0,0]$ 
      \item $[j,\alpha,R_i^1,R_i^2,c_1,c_2] \tran{\lambda}
             [j,\alpha,E_i^1,D_i^2,0,0]$ 
      \item $[j,\alpha,R_i^1,R_i^2,c_1,c_2] \tran{\vv{c_1}_1}
             [j,\alpha,R_{S(i)}^1,K^1,D_i^2,\max\{c_1{-}1,0\},0]$ 
      \item $[j,\alpha,R_i^1,R_i^2,c_1,c_2] \tran{\vv{c_2}_1}
             [j,\alpha,D_i^1,A_i^2,0,0]$ 
      \item $[j,\alpha,R_i^1,R_i^2,c_1,c_2] \tran{\vv{c_2}_2}
             [j,\alpha,D_i^1,B_i^2,0,0]$ 
      \item $[j,\alpha,R_i^1,R_i^2,c_1,c_2] \tran{\delta - \vv{c_2}_1}
             [j,\alpha,D_i^1,C_i^2,0,0]$        
      \item $[j,\alpha,R_i^1,R_i^2,c_1,c_2] \tran{\lambda}
             [j,\alpha,D_i^1,E_i^2,0,0]$ 
      \item $[j,\alpha,R_i^1,R_i^2,c_1,c_2] \tran{\vv{c_2}_1}
             [j,\alpha,D_i^2,R_{S(i)}^2,K^2,0,\max\{c_2{-}1,0\}]$ 
\end{itemize}
  where $q = 1 - \vv{c_1}_2 -\vv{c_1}_1 - \vv{c_2}_2 - \vv{c_2}_1- 2\lambda - 2\delta$.
  \item For every state $[j,\alpha,R_i^1,c_1,0]$  where $\alpha \subseteq \A^1 \cup \A^2$ is free, the immediate successors are the associated probabilities are the following:
  \begin{itemize}
        \item $[j,\alpha,R_i^1,c_1,0] \tran{\vv{c_1}_1}
               [j,\alpha,A_i^1,0,0]$ 
        \item $[j,\alpha,R_i^1,c_1,0] \tran{\vv{c_1}_2}
               [j,\alpha,B_i^1,0,0]$ 
        \item $[j,\alpha,R_i^1,c_1,0] \tran{\delta - \vv{c_1}_1}
               [j,\alpha,C_i^1,0,0]$        
        \item $[j,\alpha,R_i^1,c_1,0] \tran{q}
               [j,\alpha,D_i^1,0,0]$ 
        \item $[j,\alpha,R_i^1,c_1,0] \tran{\vv{c_1}_1}
               [j,\alpha,R_{S(i)}^1,K^1,\max\{c_1{-}1,0\},0]$ 
  \end{itemize}
  where $q = 1 - \vv{c_1}_2 - \vv{c_1}_1 - \delta$.
  \item For every state $[j,\alpha,R_i^2,0,c_2]$ where $\alpha \subseteq \A^1 \cup \A^2$ is free, the immediate successors are the associated probabilities are the following:
  \begin{itemize}
        \item $[j,\alpha,R_i^2,0,c_2] \tran{\vv{c_2}_1}
               [j,\alpha,A_i^2,0,0]$ 
        \item $[j,\alpha,R_i^2,0,c_2] \tran{\vv{c_2}_2}
               [j,\alpha,B_i^2,0,0]$ 
        \item $[j,\alpha,R_i^2,0,c_2] \tran{\delta - \vv{c_2}_1}
               [j,\alpha,C_i^2,0,0]$        
        \item $[j,\alpha,R_i^2,0,c_2] \tran{q}
               [j,\alpha,D_i^2,0,0]$ 
        \item $[j,\alpha,R_i^2,0,c_2] \tran{\vv{c_2}_1}
               [j,\alpha,R_{S(i)}^2,K^2,0,\max\{c_2{-}1,0\}]$ 
  \end{itemize}
  where $q = 1 - \vv{c_2}_2 - \vv{c_2}_1 - \delta$.
  \item If $[\alpha,c_1,c_2]$ is a state where $\alpha \subseteq \A^1 \cup \A^2$ is free, then $[\alpha,c_1,c_2] \tran{1} [\alpha,c_1,c_2]$.
\end{itemize}
A routine check confirms that $[0,r^1_{0,1},r^2_{0,1},0,0] \models \varphi_\M$ and the considered computation is modeled by the run $\pi$ such that
$\pi(j) = [j,r^1_{i,\ell},r^2_{i,\ell},c_1,c_2]$ where $\conf_j = (\ell,c_1,c_2)$ and $i = j \mbox{ mod } 3$.
\end{proof}

Let us consider the formula $\Psi_\M \equiv \xi_{\M^1} \wedge \xi_{\M^2} \wedge \Sync$ where $\xi_{\M^1}$ and  $\xi_{\M^2}$ are obtained from $\psi_{\M^1}$ and  $\psi_{\M^2}$ just by omitting the subformulae $\Rec^1$ and $\Rec^2$, respectively.
Then, one can easily show that a \emph{deterministic} two-counter machine $\M$ has a \emph{bounded} computation iff $\Psi_\M$ has a finite model. The ``$\Leftarrow$'' direction follows by using the same arguments as in Proposition~\ref{prop-correct}. The ``$\Rightarrow$'' direction uses a construction similar to the one of Proposition~\ref{prop-model}. Realize that if $\M$ is deterministic, then a bounded run $\conf_0,\conf_1,\ldots$ of $\M$ is necessarily \emph{periodic}, i.e., there exist $t,u \in \N$ such that $t < u$ and the infinite sequences $\conf_t,\conf_{t+1},\ldots$ and $\conf_u,\conf_{u+1},\ldots$ are the same. Hence, at the beginning of the construction of Proposition~\ref{prop-model}, we now add only \emph{finitely} many states of the form $[j,r^1_{i,\ell},r^2_{i,\ell},c_1,c_2]$ where $j<u$. When constructing the successors of $[u-1,r^1_{i,\ell},r^2_{i,\ell},c_1,c_2]$, we use the state $[t,r^1_{S(i),\ell'},r^2_{S(i),\ell'},c'_1,c'_2]$ instead of $[u,r^1_{S(i),\ell'},r^2_{S(i),\ell'},c'_1,c'_2]$. Otherwise, the construction is the same.
Thus, we obtain the following:

\begin{corollary}
    The \emph{finite} satisfiability problem for the $\opf,\opg$-fragment of PCTL is $\Sigma^0_1$-hard.
\end{corollary}

\end{document}